\documentclass[journal]{IEEEtran}

\usepackage{cite}
\usepackage[pdftex]{graphicx}
\usepackage{array}
\usepackage[caption=false,font=footnotesize]{subfig}
\usepackage{url}

\usepackage{amsthm,amsmath,amssymb,mathtools}
\usepackage[T1]{fontenc}
\usepackage{hyperref}
\hypersetup{colorlinks=true,citecolor=blue,linkcolor=blue}

\newtheorem{theorem}{Theorem}
\newtheorem{lemma}{Lemma}
\newtheorem{defn}{Definition}
\newtheorem{remark}{Remark}

\newcommand\norm[1]{\left\lVert#1\right\rVert}

\begin{document}

\title{Robust Attitude Tracking for Aerobatic Helicopters: A Geometric Approach}

\author{Nidhish~Raj,~%
        Ravi~N~Banavar,~%
        Abhishek,~%
        and~Mangal~Kothari% <-this % stops a space
\thanks{This work was funded by the Department of Science and Technology, India.}
\thanks{Nidhish Raj, Abhishek, and Mangal Kothari are with the Department of Aerospace Engineering, Indian Institute of Technology Kanpur, UP, 208016 India e-mail: (nraj,abhish,mangal)@iitk.ac.in.}%
\thanks{Ravi N Banavar is with the Systems and Control Engineering, Indian Institute of Technology Bombay, Mumbai, India e-mail: banavar@iitb.ac.in.}% 
}% 


\maketitle

\begin{abstract}
This paper highlights the significance of the rotor dynamics in control design for small-scale aerobatic helicopters, and proposes two singularity free robust attitude tracking controllers based on the available states for feedback. 
1. The first, employs the angular velocity and the flap angle states 
(a variable that is not easy to measure)
and uses a backstepping technique to design a robust compensator (BRC) to \textbf{\textit{actively}} suppress the disturbance induced tracking error. 2. 
The second exploits the inherent damping present in the helicopter dynamics leading to a structure preserving, \textbf{\textit{passively}} robust controller (SPR), which is free of angular velocity and flap angle feedback.
The BRC controller is designed to be robust in the presence of two types of uncertainties: structured and unstructured. The structured disturbance is due to uncertainty in the rotor parameters, and the unstructured perturbation is modeled as an exogenous torque acting on the fuselage. The performance of the controller is demonstrated in the presence of both types of disturbances through numerical simulations. 
In contrast, the SPR tracking controller is derived such that the tracking error dynamics inherits the natural damping characteristic of the helicopter. 
%This avoids the unnecessary cancellation of inherent damping, thereby making it robust. 
The SPR controller is shown to be almost globally asymptotically stable and its performance is evaluated experimentally by performing aggressive flip maneuvers. Throughout the study, a nonlinear coupled rotor-fuselage helicopter model with first order flap dynamics is used.
\end{abstract}


\section{Introduction} \label{sec1}
\IEEEPARstart{S}{mall-scale} helicopters with a single main rotor and a tail rotor are capable of performing extreme 3D aerobatic maneuvers \cite{gavrilets2001aggressive,abbeel2010autonomous,gerig2008modeling}. 
These aggressive maneuvers involve large angle rotations with high angular velocity, inverted flight, split-S, pirouette etc. This necessitates a robust attitude tracking controller which is globally defined and is capable of achieving such aggressive rotational maneuvers.

The attitude tracking problem of a helicopter is significantly different from that of a rigid body. A small-scale helicopter is modeled as a coupled interconnected system consisting of a fuselage and a rotor. The control moments generated by the rotor excite the rigid body dynamics of the fuselage which in-turn affects the rotor loads and its dynamics causing nonlinear coupling. The key differences between the rigid body tracking problem and the attitude tracking of a helicopter are the following: 1) the presence of large aerodynamic damping in the rotational dynamics and 2) the required control moment for tracking cannot be applied instantaneously due to the rotor blade dynamics. The control moments are produced by the rotor subsystem which, for the purpose of attitude tracking, can be approximated as a first order system \cite{mettler2013identification}. The damping introduced by the rotor subsystem does not hamper attitude stabilization or slow trajectory tracking. But it becomes a serious impediment for fast aerobatic maneuvers, which is what has been addressed in this article. This is in contrast to a quadrotor, where the rigid body approximation is close to the actual dynamics due to relatively small rotors.

\subsection{Related Work}
The significance of including the rotor dynamics in controller design for helicopters has been extensively studied in the literature \cite{hall1973inclusion,takahashi1994h,ingle1994effects,panza2014rotor}.  Hall Jr and Bryson Jr \cite{hall1973inclusion} have shown the importance of rotor state feedback in achieving tight attitude control for large scale helicopters, Takahashi \cite{takahashi1994h} compares $H_\infty$ attitude controller design for the cases with and without rotor state feedback. In a similar work, Ingle and Celi \cite{ingle1994effects} have investigated the effect of including rotor dynamics on various controllers, namely LQG, Eigenstructure Assignment and $H_\infty$, for meeting stringent handling quality requirements. They conclude that the controllers designed to meet the high bandwidth requirements with the rotor dynamics were more robust and required lower control activity than the ones designed without including the rotor dynamics. Panza and Lovera \cite{panza2014rotor} used rotor state feedback and designed an $H_\infty$ controller which is robust and also fault tolerant with respect to failure of the rotor state sensor. Previous attempts to small-scale helicopter attitude control are mostly based on attitude parametrization such as Euler angles, which suffer from  singularity issues, or quaternions which have ambiguity in representation. Tang, Yang, Qian, and Zheng \cite{tang2015attitude} explicitly consider the rotor dynamics and design stabilizing controller based on sliding mode technique using Euler angles and hence confined to small angle maneuvers. Raptis, Valavanis, and Moreno \cite{raptis2011novel} have designed position tracking controller for small-scale helicopter wherein the inner loop attitude controller was based on rotation matrix, but does not consider the rotor dynamics. Marconi and Naldi \cite{marconi2006robust,marconi2007robust} designed a position tracking controller for flybared (with stabilizer bar) miniature helicopter which is robust with respect to large variations in parameters, but have made the simplifying assumption of disregarding the rotor dynamics by taking a static relation between the flap angles and the cyclic input. Stressing the significance of rotor dynamics,  Ahmed and Pota \cite{ahmed2009flight,ahmed2010flight} developed a backstepping based stabilizing controller using Euler angles for a small-scale flybared helicopter with the inclusion of servo and rotor dynamics. They have provided correction terms in the controller to incorporate the effect of servo and rotor dynamics. For near hover conditions, Zhu and Huo \cite{zhu2013robust} have developed a robust nonlinear controller disregarding the flap dynamics.  Frazzoli, Dahleh, and Feron \cite{frazzoli2000trajectory} developed a coordinate chart independent trajectory tracking controller on the configuration manifold $SE(3)$ for a small-scale helicopter, but the flap dynamics, critical for accurate representation of the system, was not taken into account.

\subsection{Contribution}
The present work emphasizes on the inclusion of the rotor dynamics in the design of attitude tracking controller for small scale aerobatic helicopters. It is an improvement over the simple attitude tracking controller proposed by the authors in \cite{raj2017attitude}. The backstepping technique employed in the previous work warranted the removal of damping term from the dynamics for performing aggressive maneuvers. As shown in Sec \ref{sec:sim}, uncertainties in the rotor time constant estimate, $\tau_m$, or its variation with vehicle operating condition could result in excess removal of damping, thereby injecting energy into the system and making the closed loop unstable. The novelty of the present work is in the design of two singularity free attitude tracking controllers which are robust to the aforementioned parametric variation. 
The first controller uses backstepping technique to design a robust compensator to actively suppress the disturbance induced tracking error. Disturbance entering into the system due to uncertainty in parameters is termed \textbf{\textit{structured}} uncertainty. In addition, the effect of a time varying exogenous torque acting on the fuselage is also considered and is lumped together as $\Delta_f(t)$ in \eqref{eq:rigid_body}, hence termed \textbf{\textit{unstructured}}.
The proposed BRC controller ensures robustness to both uncertainties and renders the solutions of the associated error dynamics to be uniformly ultimately bounded. It is observed through numerical simulations that the proposed controller is capable of performing aggressive rotational maneuvers in the presence of the aforementioned structured and unstructured disturbances. The disadvantage of this controller is the need for, difficult to measure, flap angle state for feedback.
On the other hand, the second controller is derived such that the tracking error dynamics inherits the natural damping characteristic of the helicopter. This avoids the unnecessary cancellation of inherent damping, thereby making it robust. This controller has the added advantage of being free of flap and angular velocity feedback and hence easily implementable. On the downside, due to the passive nature of its robustness, the controller cannot confine the tracking error to prescribed limits. The performance of the controller is proven through experiments by performing aggressive flip maneuvers on a small scale helicopter.

%\subsection{Outline}
The paper is organized as follows: Section \ref{sec:model} describes the rotor-fuselage dynamics of a small-scale helicopter, explains the effect of aerodynamic damping, and motivates the need for a robust attitude tracking controller. Section \ref{sec:rob_cont} presents the backstepping based robust attitude tracking controller. The efficacy of the BRC controller is demonstrated through numerical simulation in Section \ref{sec:sim}. The structure preserving controller and its experimental validation are provided in Sections \ref{sec:str_rob_cont} and \ref{sec:Exp} respectively.

%%%%%%%%%%%%%%%%%%%%%%%%%%%%%%%%%%%%%%%%%%%%%%%%%%%%%%%%%%%%%%%%%%%%%%%%%
%%%%%%%%%%%%%%%%%%%%%%%%%%%%%%%%%%%%%%%%%%%%%%%%%%%%%%%%%%%%%%%%%%%%%%%%%
%%%%%%%%%%%%%%%%%%%%%%%%%%%%%%%%%%%%%%%%%%%%%%%%%%%%%%%%%%%%%%%%%%%%%%%%%

\section{Helicopter Model} \label{sec:model}

Unlike quadrotors, a helicopter modeled as a rigid body does not capture all the dynamics required for designing high bandwidth attitude tracking controllers. For this purpose, we consider here a minimal model of a small scale flybarless helicopter which consists of a fuselage and a rotor. The fuselage is modeled as a rigid body and the rotor as a first order system which generates the required control moment. The inclusion of the rotor model is crucial, since it introduces the significant aerodynamic damping in the system, which is an integral part of the dynamics of the helicopter. This clearly distinguishes the helicopter control problem from that of rigid bodies in space and robotics applications, where this interplay is not present.

The rotational equations of motion of the fuselage are given by,
\begin{equation} \label{eq:rigid_body}
\begin{split}
\dot{R} &= R\hat{\omega}, \\
J\dot{\omega} + \omega\times J\omega &= M + \Delta_f(t),
\end{split}
\end{equation}
where $R \in SO(3)$ is the rotation matrix which transforms vectors from the body fixed frame, $(O_b,X_b,Y_b,Z_b)$, to the inertial frame of reference, $(O_e,X_e,Y_e,Z_e)$. $M \triangleq [M_x,M_y,M_z]$ is the external moment acting on the fuselage due to the rotor, $\Delta_f$ is a time varying disturbance torque bounded by $\norm{\Delta_f(t)} < \delta_f$. This torque could arise due to center of mass offset from the main rotor shaft axis or the moment due to a slung load attached to a point different from the center of mass. $J$ is the body moment of inertia of the fuselage, and $\omega = [\omega_x,\omega_y,\omega_z]$ is the angular velocity of the fuselage expressed in the body frame. The hat operator, $\hat{(\cdot)}$, is a Lie algebra isomorphism from $\mathbb{R}^3$ to $\mathfrak{so(3)}$ given by
\begin{equation*}
\hat{\omega} = \begin{bmatrix}
0 & -\omega_z & \omega_y \\
\omega_z & 0 & -\omega_x \\
-\omega_y & \omega_x & 0 
\end{bmatrix}.
\end{equation*}

\begin{figure}[!t]
\centering
\captionsetup{justification=centering}
\includegraphics[scale=0.75]{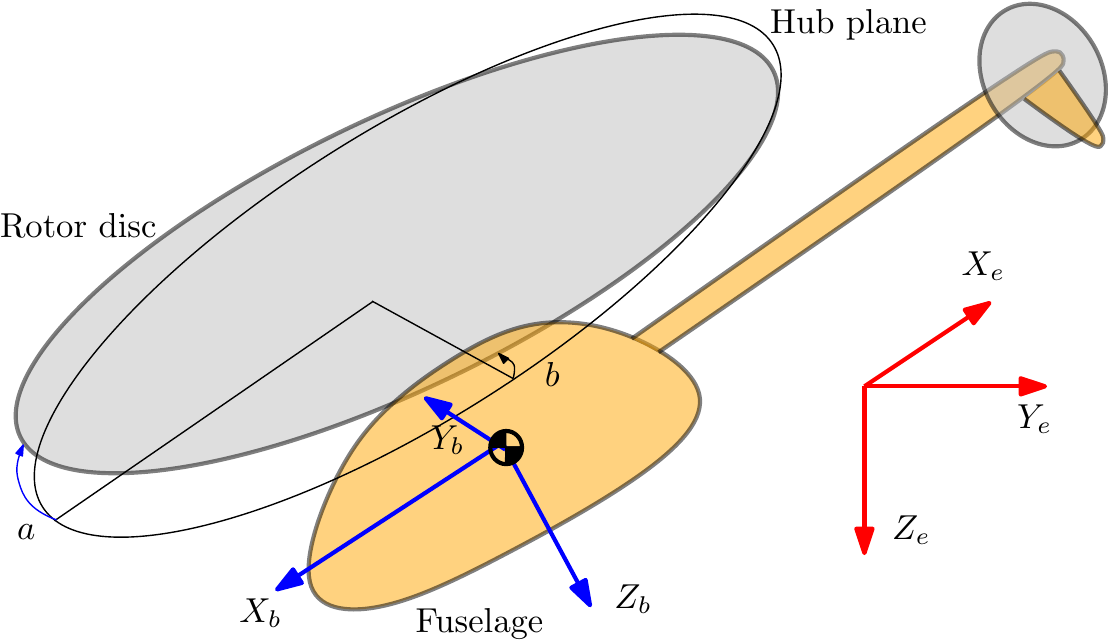}
\caption{Fuselage and rotor disc with flap angles.}
\label{fig:heli_model}
\end{figure}

The first order tip path plane (TPP) equations for the rotor are considered, as they capture the required dynamics for gross movement of the fuselage \cite{mettler2013identification}. The coupled flap equation for a counter-clockwise rotor are given by \cite{mettler2013identification, chen1980effects}
\begin{equation} \label{eq:flap_eq}
\begin{split}
\dot{a} = - \frac{1}{\tau_m}a + \frac{k_\beta}{2 \Omega I_\beta}b -\omega_y  + \frac{1}{\tau_m} \left( \theta_a - \frac{\omega_x}{\Omega} \right) ,\\
\dot{b} = - \frac{1}{\tau_m}b - \frac{k_\beta}{2 \Omega I_\beta}a -\omega_x  + \frac{1}{\tau_m}\left( \theta_b + \frac{\omega_y}{\Omega} \right). 
\end{split}
\end{equation}
where $a$ and $b$ are respectively the longitudinal and lateral tilt of the rotor disc with respect to the hub plane as shown in Fig. \ref{fig:heli_model}. $\tau_m$ is the main rotor time constant and $\theta_a$ and $\theta_b$ are the control inputs to the rotor subsystem. They are respectively the lateral and longitudinal cyclic blade pitch angles actuated by servos through a swash plate mechanism. $k_\beta$ is the blade root stiffness, $\Omega$ is the main rotor angular velocity, and $I_\beta$ is the blade moment of inertia about the flap hinge.
The above equation introduces cross coupling through flap angle and angular velocity. 
Note that the effect of the angular velocity cross coupling can be effectively canceled using the fuselage angular velocity feedback, since the rotor angular velocity, $\Omega$, can be measured accurately using an on-board autopilot.

%\begin{figure}[!htbp]
%\begin{center}
%\includegraphics[scale=0.75]{figures/heli_tpp.pdf}
%%\includegraphics[scale=0.75]{figures/heli_frame.pdf}
%\caption{Rotor-fuselage coupling and flap angles.}
%\label{fig:rotor_fuselage}
%\end{center}
%\end{figure}

The coupling of rotor and fuselage occurs through the rotor hub. The rolling moment, $M_x$ and pitching moment $M_y$, acting on the fuselage due to the rotor flapping consists of two components -- due to tilting of the thrust vector, $T$, and due to the rotor hub stiffness, $k_\beta$, and are given by
\begin{equation}
\begin{split}
M_x &= (hT + k_\beta)b, \\
M_y &= (hT + k_\beta)a,
\end{split}
\end{equation}
where $h$ is the distance of the rotor hub from the center of mass. For small-scale helicopters, the rotor hub stiffness $k_\beta$ is much larger than the component due to tilting of thrust vector, $hT$ (see Table \ref{tab:heli_params}). Thus, a nominal variation in thrust would result in only a small variation of the equivalent hub stiffness, $K_\beta \triangleq (hT + k_\beta)$. The control moment about yaw axis, $M_z$, is applied through tail rotor which, due to it's higher RPM, has a much faster aerodynamic response than the main rotor flap dynamics. The tail rotor along with the actuating servo is modeled as a first order system with $\tau_t$ as the tail rotor time constant,
\begin{equation}
\dot{M}_z = -M_z/\tau_t -K_t \omega_z  + K_t K_{t0}\theta_t/\tau_t,
\end{equation}
where $K_{t0}$ relates the steady state yaw rate to control input, $\theta_t$.

\begin{figure} [!t]
\centering
\subfloat[]{\includegraphics[width=1\linewidth,height=8cm]{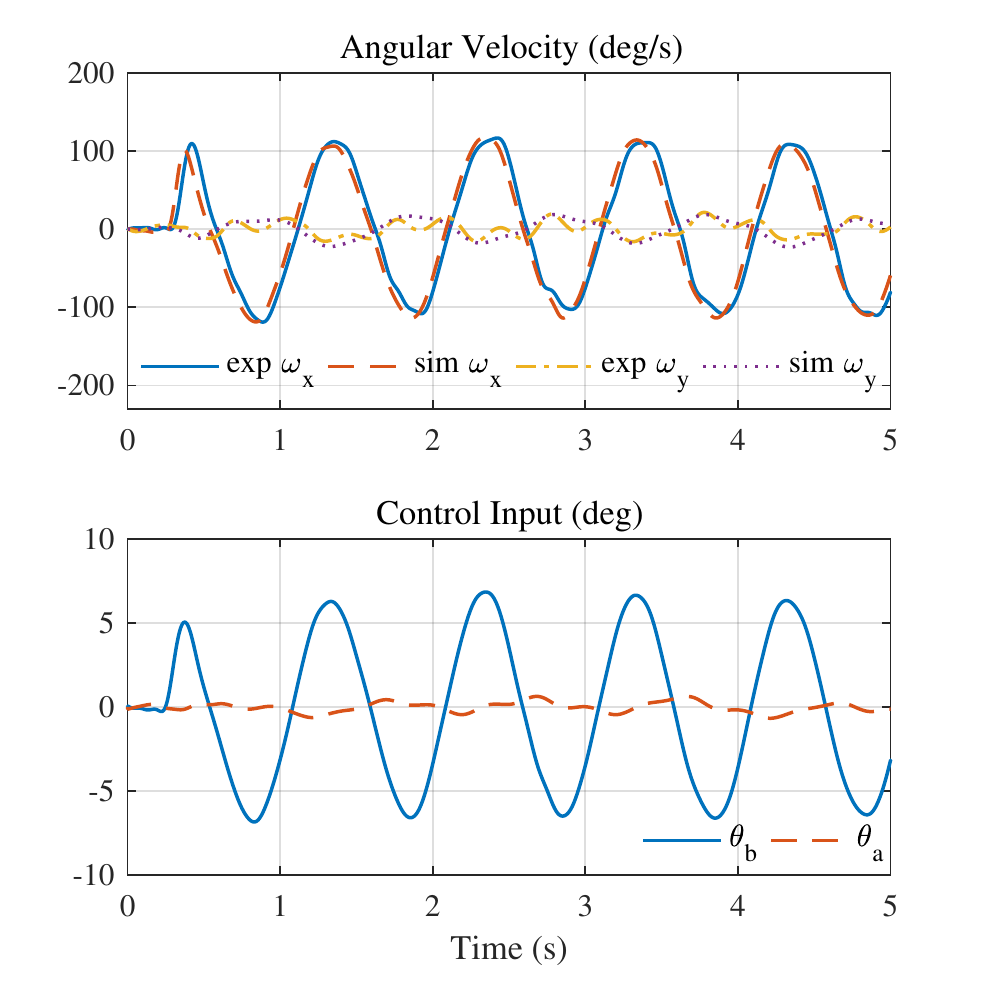}}\\
\subfloat[]{\includegraphics[width=1\linewidth,height=8cm]{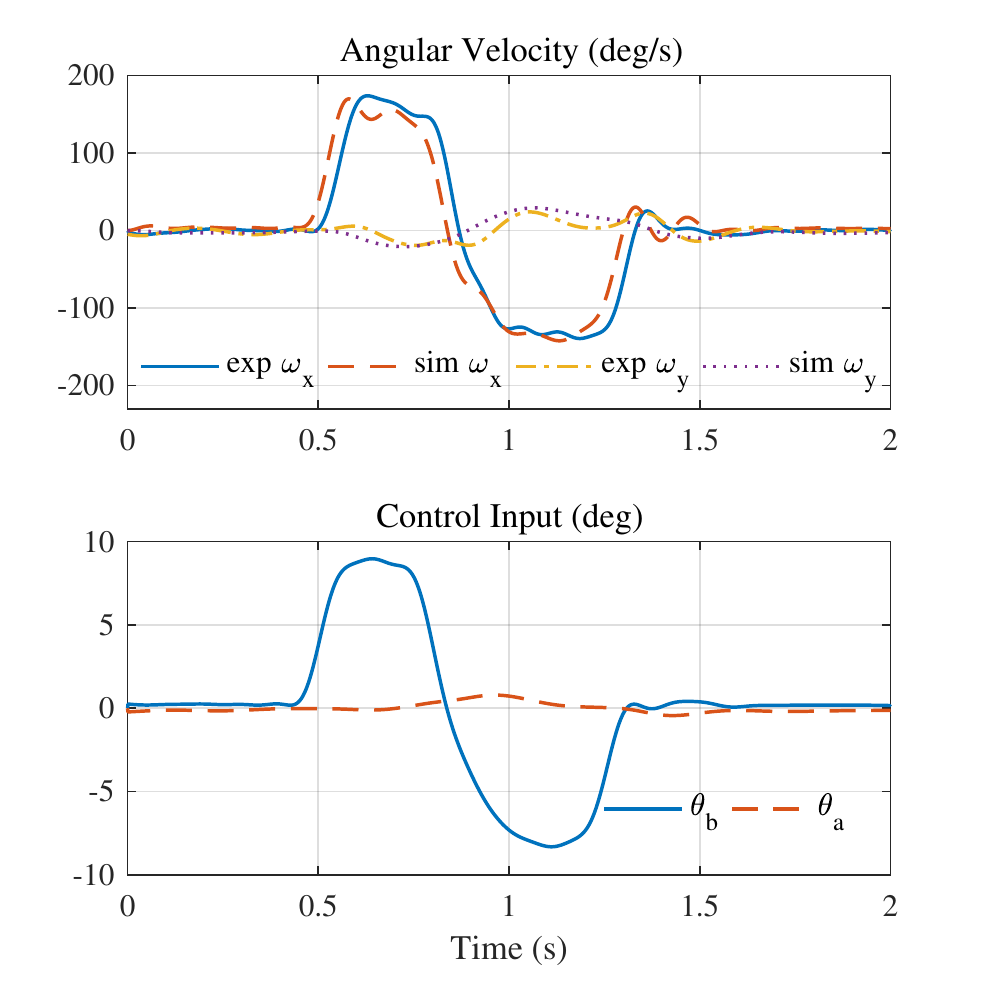}}
\caption{Model validation with experimental flight data. (a) Sinusoidal input, (b) Doublet input}
\label{fig:model_val}
\end{figure}

The main rotor dynamics \eqref{eq:flap_eq} and tail rotor dynamics could be written in terms of the control moments and a pseudo-control input $\theta \triangleq [\theta_b + \omega_y/\Omega, \theta_a - \omega_x/\Omega, K_{t0}\theta_t]$ as,
\begin{equation} \label{eq:actuator_th}
\dot{M} = AM - K\omega + KA_\tau\theta,
\end{equation}
where
\begin{equation}
A \triangleq \begin{bmatrix}
-1/\tau_m & -k & 0 \\
 k & -1/\tau_m & 0 \\
0 & 0 & -1/\tau_t 
\end{bmatrix},
\end{equation}
$A_\tau \triangleq diag(1/\tau_m,1/\tau_m,1/\tau_t)$, $K \triangleq diag(K_\beta,K_\beta,K_t) $, $k \triangleq k_\beta/(2\Omega I_\beta)$, and $M \triangleq [M_x,M_y,M_z]$. The symmetric and skew symmetric parts of $A$ are denoted by $-A_\tau$ and $A_k$ respectively.
Note that the combined rotor-fuselage dynamics given by \eqref{eq:rigid_body} and \eqref{eq:actuator_th}  cannot be given the form of a \textit{simple mechanical system} \cite{bullo2004geometric} as the actuator dynamics is first order. This precludes the equations of motion being written in the usual form of a geodesic on a Riemannian manifold and standard control techniques for such systems being employed.

The above model is validated for a small scale aerobatic helicopter (Align Trex 700), used for both simulation and experimentation in later sections. Figure \ref{fig:model_val} shows the close match between the experimental data and simulated model for the same input. The vehicle parameters were estimated with sinusoidal input and validated for doublet input response and are given in Table \ref{tab:heli_params}. 

\begin{figure} [!t] 
\centering
\includegraphics[width=1\linewidth]{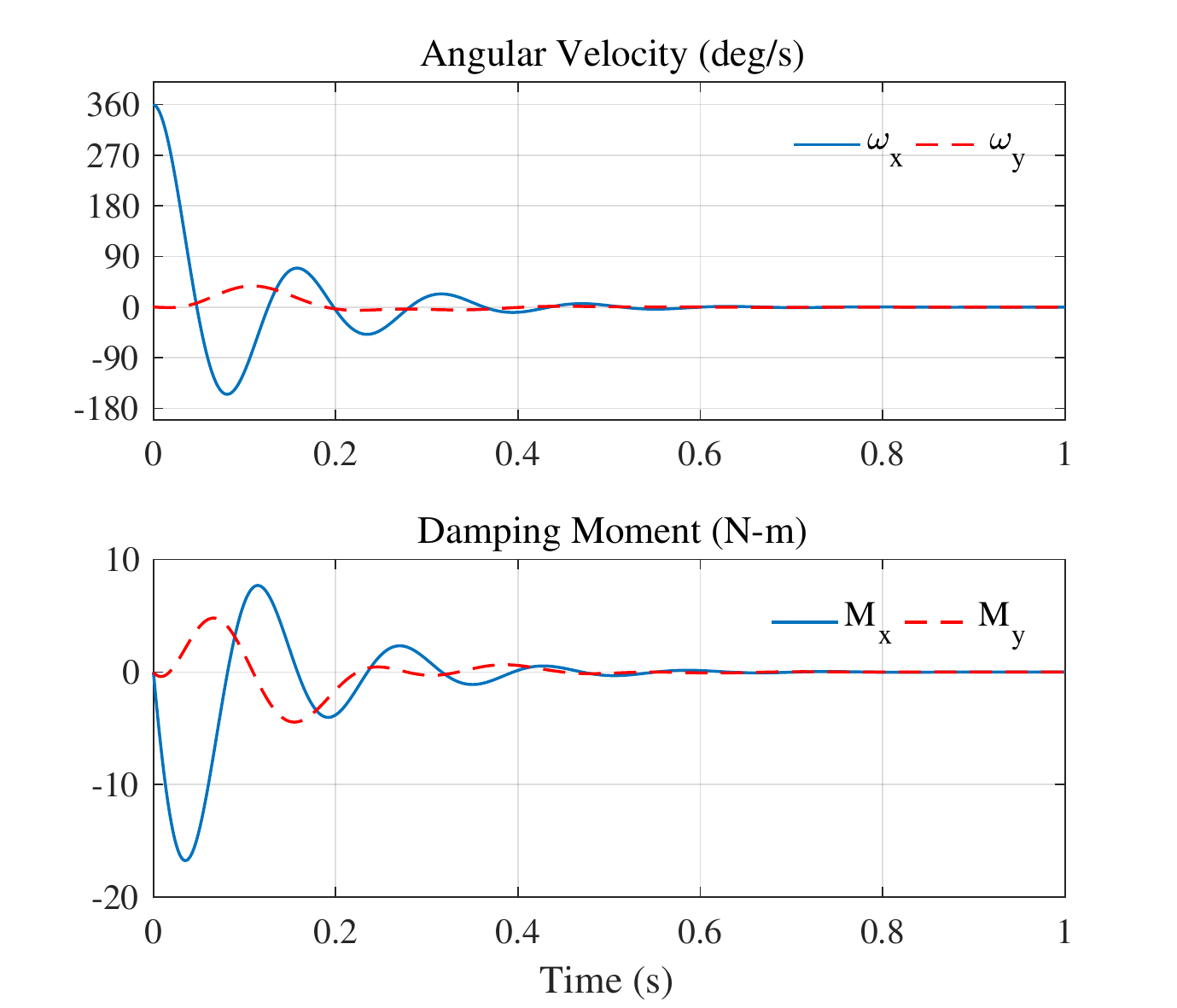}
\caption{The effect of aerodynamic damping on attitude dynamics of small-scale helicopter. Response to an initial condition of 360 deg/s about roll axis. Top: Angular Velocity, Bottom: Damping Moment generated}
\label{fig:damping}
\end{figure}

%\begin{figure} [!t] 
%\centering
%\subfloat[]{\includegraphics[width=1\linewidth,height=4cm]{figures/damping_roll.pdf}}\\
%\subfloat[]{\includegraphics[width=1\linewidth,height=4cm]{figures/damping_moment.pdf}}
%\caption{The effect of aerodynamic damping on attitude dynamics of small-scale helicopter. Response to an initial condition of 360 deg/s about roll axis (a) Angular Velocity, (b)Damping Moment generated}
%\label{fig:damping}
%\end{figure}

We now elaborate a particular physical feature of the dynamics. The aerodynamic damping comes in through the presence of the negative angular velocity term, $-K\omega$, in \eqref{eq:actuator_th}. A positive angular velocity builds up negative flap angles, leading to negative moment and vice versa. The effect of the damping on the attitude dynamics of a small-scale helicopter (parameters in Table \ref{tab:heli_params}) is depicted in Fig \ref{fig:damping}. An initial angular velocity of 360 deg/s about the roll axis is damped to zero in less than a second, a characteristic which is absent in satellites and negligible in quadrotors. The maximum damping moment generated in the process is 17 N-m, which is a significant torque for a rigid body of this size. Since, aerobatic maneuvers involve attitude trajectories of large angular velocities, there is a need to design controllers which either actively cancels the damping effect, as done in the case of BRC controller, or takes advantage of it, as in the case of SPR controller. The effect of damping can be canceled effectively by augmenting the control input with angular velocity feedback. Exact cancellation requires the knowledge of the true value of the main rotor time constant, $\tau_m$. An overestimate of $\tau_m$ if used for canceling $\omega$ would result in a positive feedback of the angular velocity, which could make the attitude dynamics unstable as shown through simulation in Fig \ref{fig:str_fail}. The parameters of the rotor dynamics, rotor stiffness $K_\beta$, $K_t$ and time constants $\tau_m$, $\tau_t$, are estimated from flight data by means of system identification and hence prone to be erroneous. But we restrict our analysis to the case where the rotor stiffness parameters are known perfectly, although it can be easily extended to include such uncertainties. This assumption is justified as it has been observed through numerical simulations that even a significant error in K does not lead to tracking instability. This motivates the backstepping controller proposed in the paper which is robust with respect to uncertainty in the main rotor time constant. On the contrary, the structure preserving controller preserves the damping to achieve the same goal.

The estimates of rotor time constants are given by $\bar{\tau}_m = (1+\alpha_m)\tau_m$ and $\bar{\tau}_t = (1+\alpha_t)\tau_t$ where the uncertainty parameters $\alpha_m$, $\alpha_t$ satisfy $|\alpha_m|<\alpha_{m,max}<1$ and $|\alpha_t|<\alpha_{t,max}<1$. Define the maximum variation in parameters to be $\alpha \triangleq \max\{\alpha_{m,max},\alpha_{t,max}\}$. The estimates of $A_\tau$ and $A$ are given by
\begin{equation}
\bar{A}_\tau \triangleq  \begin{bmatrix}
1/\bar{\tau}_m & 0 & 0 \\
0 & 1/\bar{\tau}_m & 0 \\
0 & 0 & 1/\bar{\tau}_t 
\end{bmatrix}, \quad \bar{A} = -\bar{A}_\tau + A_k.
\end{equation}
Using the above relations we have
\begin{equation} \label{eq:A_norm}
\begin{gathered}
A_\tau \bar{A}_\tau^{-1} =  \begin{bmatrix}
1+\alpha_m & 0 & 0 \\
0 & 1+\alpha_m & 0 \\
0 & 0 & 1+\alpha_t 
\end{bmatrix}, \\
\max_{\substack{\alpha_m, \alpha_t}} \norm{I-A_\tau\bar{A}_\tau^{-1}} = \alpha.
\end{gathered}
\end{equation}

%%%%%%%%%%%%%%%%%%%%%%%%%%%%%%%%%%%%%%%%%%%%%%%%%%%%%%%%%%%%%%%%%%%%%%%%%%%%%%%
%%%%%%%%%%%%%%%%%%%%%%%%%%%%%%%%%%%%%%%%%%%%%%%%%%%%%%%%%%%%%%%%%%%%%%%%%%%%%%%
%%%%%%%%%%%%%%%%%%%%%%%%%%%%%%%%%%%%%%%%%%%%%%%%%%%%%%%%%%%%%%%%%%%%%%%%%%%%%%%

\section{Backstepping Robust Controller} \label{sec:rob_cont}

Given a twice differentiable attitude reference command $(R_d(t),\omega_d(t),\dot{\omega}_d(t))$, the objective is to design an attitude tracking controller for the helicopter. The combined rotor-fuselage dynamics is reproduced here for convenience
\begin{subequations} \label{eq:rotor-fuse}
\begin{equation}
\dot{R} = R\hat{\omega}, 
\end{equation}
\begin{equation} \label{eq:fuselage}
J\dot{\omega} + \omega\times J\omega = M + \Delta_f(t), 
\end{equation}
\begin{equation} \label{eq:rotor}
\dot{M} = AM - K\omega + KA_\tau \theta.
\end{equation}
\end{subequations}

In order to apply the backstepping approach, first the tracking problem is transformed to stabilization of the error dynamics. Then, a robust attitude tracking controller for the fuselage subsystem is designed, as is described in \cite{lee2013nonlinear}, which is subsequently extended to include the rotor dynamics. The configuration space of fuselage subsystem is the Lie group $SO(3)$. As is pointed out by Maithripala, Berg, and Dayawansa \cite{maithripala2006almost}, a general tracking problem on a Lie group can be reduced to a configuration stabilization problem about the identity element of the group. This is possible on a Lie group since the error between any two configurations can be naturally defined using the group operation. Such an operation is not always defined on a general configuration manifold. For the set of rotation matrices, $SO(3)$, the rotation error matrix between the current rotation $R$ and the desired rotation $R_d$ is defined as $R_e \triangleq R_d^T R$. $R_e$ transforms a vector from the current body frame to the desired body frame. To obtain the fuselage error dynamics, differentiate $R_e$
\begin{equation}
\dot{R}_e = R_d^T R \hat{\omega} - \hat{\omega}_d R_d^T R = R_e(\hat{\omega} - R_e^T \hat{\omega}_d R_e). 
\end{equation} 
Defining $e_\omega \triangleq \omega - R_e^T \omega_d$, then $\dot{R}_e$ reduces to
\begin{equation} \label{eq:Re_dot}
\dot{R}_e = R_e\hat{e}_\omega.
\end{equation}
Differentiate $J e_\omega$ to obtain the fuselage error dynamics
\begin{equation} \label{eq:att_err_dyn}
J \dot{e}_\omega = -\omega \times J\omega + J(\hat{e}_\omega R_e^T \omega_d - R_e^T \dot{\omega}_d) + M + \Delta_f.
\end{equation}
Defining $e_M \triangleq M - M_d$ as the difference between the actual and desired moment acting on the fuselage due to the rotor.
Differentiating $e_M$ gives the rotor error dynamics,
\begin{equation} \label{eq:rotor_err_dyn}
\dot{e}_M = Ae_M + AM_d - \dot{M}_d - K\omega + KA_\tau \theta.
\end{equation}
The equilibrium of the error dynamics corresponding to zero tracking error and hence meant to be stabilized is $(R_e,e_\omega,e_M) = (I,0,0)$. The tracking error stabilizing controller has proportional derivative plus feed-forward components. The proportional action is derived from a tracking error function $\psi : SO(3) \to \mathbb{R}$ which is defined as
\begin{equation}
\psi(R_e) \triangleq \frac{1}{2} tr[I - R_e].
\end{equation}
$\psi$ has a single critical point within the sub level set about the identity I, $\Gamma_{2} \triangleq \{ R \in SO(3) | \psi(R) < 2 \}$. This sublevel set represents the set of all rotations which are less than $\pi$ radians from the identity $I$. The derivative of $\psi$ is given by 
\begin{equation}
\begin{split}
\frac{d}{dt} \psi(R_e(t)) = \frac{1}{2} tr(-\dot{R}_e(t)) = -\frac{1}{2} tr(R_e \hat{e}_\omega) \\
= -\frac{1}{2} tr \left( \frac{1}{2} (R_e-R_e^T) \hat{e}_\omega \right) = e_R \cdot e_\omega,
\end{split}
\end{equation}
where the rotation error vector is
\begin{equation} \label{eq:prop_cont}
e_R \triangleq \frac{1}{2}[R_e-R_e^T]^\vee,
\end{equation}
where $(\cdot)^\vee : \mathfrak{so(3)}\rightarrow\mathbb{R}^3$ is the inverse of hat map $\hat{(\cdot)}$. The above derivation uses the fact that $-\frac{1}{2} tr(\hat{a}\hat{b}) = a \cdot b$ and the trace of the product of symmetric and skew symmetric matrices is zero.
The total derivative of $e_R$ is
\begin{equation} \label{eq:eR_dot}
\begin{split}
\dot{e}_R &= \frac{1}{2}(\dot{R}_e - \dot{R}_e^T)^\vee = \frac{1}{2}(R_e \hat{e}_\omega  + \hat{e}_\omega R_e^T)^\vee  \\
&= B(R_e)e_\omega,
\end{split}
\end{equation}
where $B(R_e) \triangleq \frac{1}{2}[tr(R_e^T)I - R_e^T]$.
Since $\psi$ is positive definite and quadratic within the sub level set $\Gamma_{\xi_2} \triangleq \{ R \in SO(3) | \psi(R) \leq \xi_2 \}$ for some positive $\xi_2 <2$, this makes $\psi$ uniformly quadratic about the identity \cite{bullo2004geometric}. This implies there exists positive constants $b_1 = 1/2$ and $b_2 = 1/(2-\xi_2)$ \cite{lee2011robust_adaptive} such that 
\begin{equation}
b_1\norm{e_R}^2 \leq \psi(R) \leq b_2\norm{e_R}^2.
\end{equation}

%%%%%%%%%%%%%%%%%%%%%%%%%%%%%%%%%%%%%%%%%%%%%%%%%%%%%%%%%%%%%%%%%%%%%%%%%%%%%%%%%%%%%%%%%%
The following definition of ultimate boundedness taken from \cite{khalil2002nonlinear} has been given here for the sake of completeness.
\begin{defn}
Consider the system
\begin{equation} \label{eq:lemma1_sys}
\dot{x} = f(t,x)
\end{equation}
where $f:[0,\infty) \times D \to \mathbb{R}^n$ is piecewise continuous in $t$ and locally Lipschitz in $x$ on $[0,\infty) \times D$, and $D \subset \mathbb{R}^n$ contains the origin (equilibrium).
The solutions of \eqref{eq:lemma1_sys} are \textbf{uniformly ultimately bounded} with ultimate bound $\mathtt{b}$ if there exist positive constants $\mathtt{b}$ and $\mathtt{c}$, independent of $t_0$, and for every $\mathtt{a} \in  (0,\mathtt{c})$, there is $T = T(\mathtt{a},\mathtt{b}) \geq 0$, independent of $t_0$, such that
\begin{equation}
\norm{x(t_0)} \leq \mathtt{a} \implies \norm{x(t)} \leq \mathtt{b}, \quad \forall t \geq t_0 + T.
\end{equation}  
\end{defn}
The following lemma uses Lyapunov analysis to show ultimate boundedness for \eqref{eq:lemma1_sys} and is a variation of Theorem 4.18 in \cite{khalil2002nonlinear}.
\begin{lemma}(\cite{khalil2002nonlinear})
Let $D \subset \mathbb{R}^n$ be a domain that contains the origin and $V:[0,\infty) \times D \to \mathbb{R}$ be a continuously differentiable function such that
\begin{equation*} \label{eq:lemma1_condt}
\begin{gathered}
k_1 \norm{x}^2 \leq V(t,x) \leq k_2 \norm{x}^2, \quad
\dot{V} \leq -k_3 \norm{x}^2 \\
\forall x \in \Lambda_{c_1}^{c_2} \triangleq \{ x \in D | c_1 \leq V \leq c_2 , 0<c_1<c_2 \}, \quad \forall t \geq 0,
\end{gathered}
\end{equation*}
for positive $k_1, k_2, k_3$ and consider the sublevel set $L_{c_2}^- \triangleq \{x \in \mathbb{R}^n | V \leq c_2 \} \subset D$. Then, for every initial condition $x(t_0) \in L_{c_2}^- $, the solution of \eqref{eq:lemma1_sys} is uniformly ultimately bounded with ultimate bound $\mathtt{b}$, i.e. there exists $T \geq 0$ such that
\begin{subequations}
\begin{equation} \label{eq:lemma1_result1}
\norm{x(t)} \leq \left( \frac{k_2}{k_1} \right)^{1/2} \norm{x(t_0)} e^{-\gamma (t-t_0)}, \quad \forall t_0 \leq t \leq t_0 + T
\end{equation}
\begin{equation} \label{eq:lemma1_result2}
\norm{x(t)} \leq \mathtt{b} , \quad \forall  t \geq t_0 + T 
\end{equation}
\end{subequations}
where $\gamma = \frac{k_3}{2k_2}$, $\mathtt{b} = \left( \frac{c_1}{k_1}\right)^{1/2}$.
\end{lemma}
\noindent
The proof is straightforward. The following theorem presents the backstepping robust controller.

\begin{theorem}
For all initial conditions starting in the set $S \triangleq \{(R_e,e_\omega,e_M)\in SO(3)\times \mathbb{R}^3 \times \mathbb{R}^3  | \psi(R_e) + \frac{1}{2} \tilde{e}_\omega \cdot J \tilde{e}_\omega + \frac{1}{2}e_M \cdot e_M \leq \xi_2 \}$ for a positive $\xi_2 < 2$, the solutions of the error dynamics \eqref{eq:Re_dot}, \eqref{eq:att_err_dyn}, and \eqref{eq:rotor_err_dyn} are rendered uniformly ultimately bounded  by the following choice of control input 
\begin{equation} \label{eq:control_input}
\theta = (K\bar{A}_\tau)^{-1}(-\bar{A}M_d + \dot{M}_d -\tilde{e}_\omega + K\omega + \mu_r),
\end{equation}
where $\tilde{e}_\omega = e_\omega + k_R e_R$,
\begin{equation} \label{eq:moment_des}
\begin{gathered}
M_d = -k_\omega \tilde{e}_\omega - e_R - k_R J B e_\omega + \omega \times J \omega
\\ -J(\hat{e}_\omega R_e^T \omega_d - R_e^T \dot{\omega}_d) + \mu_f, \\
\mu_f = \frac{-\delta_f^2 \tilde{e}_\omega}{\delta_f \norm{\tilde{e}_\omega} + \epsilon_f},
\end{gathered}
\end{equation}
\begin{equation} \label{eq:temp1}
\begin{gathered}
\mu_r = \frac{-\alpha}{1-\alpha} \frac{\norm{\delta_r}^2 e_M}{\norm{\delta_r}\norm{e_M} + \epsilon_r}, \\
\delta_r = \tilde{e}_\omega + A_kM_d - \dot{M}_d - K\omega,
\end{gathered}
\end{equation}
for some $k_R>0$, $k_\omega >0$ and $\epsilon_f>0$, $\epsilon_r>0$ such that
\begin{equation} \label{eq:epsilon_condt}
\epsilon \triangleq \epsilon_f + \epsilon_r < \xi_2\frac{\lambda_{min}(W)}{\lambda_{max}(U_2)}.
\end{equation}
The ultimate bound is given by 
\begin{equation} \label{eq:ultimate_bound}
\mathtt{b} = \left( \frac{\lambda_{max}(U_2)}{\lambda_{min}(U_1)\lambda_{min}(W)} \epsilon \right)^{1/2} .
\end{equation}
The matrices $U_1$, $U_2$ and $W$ are given by
\begin{equation} \label{eq:matricesU1U2W}
\begin{gathered}
U_1 \triangleq \frac{1}{2}\begin{bmatrix}
1 & 0 & 0 \\
0 & \lambda_{min}(J) & 0 \\
0 & 0 & 1 
\end{bmatrix},
U_2 \triangleq \frac{1}{2}\begin{bmatrix}
\frac{2}{2-\xi_2} & 0 & 0 \\
0 & \lambda_{max}(J) & 0 \\
0 & 0 & 1 
\end{bmatrix}, \\
W \triangleq \begin{bmatrix}
 k_R & 0 & 0 \\
 0 & k_\omega & 0 \\
 0 & 0 & \lambda_{min}(A_\tau) 
\end{bmatrix}.
\end{gathered}
\end{equation}
\end{theorem}

\begin{proof}

Consider the following positive definite quadratic function in the sublevel set $\Gamma_{\xi_2}$, $V_1 \triangleq \psi $. The time derivative of this function, $\dot{V}_1 = e_R \cdot e_\omega $, can be made negative definite by setting the virtual control input $e_\omega = -k_R e_R$. A change of variable $\tilde{e}_\omega = e_\omega + k_R e_R$ would make $\dot{V}_1 = -k_R \norm{e_R}^2 + e_R \cdot \tilde{e}_\omega $.\\
The error dynamics for $\tilde{e}_\omega$ is given by
\begin{equation*}
\begin{split}
\dot{\tilde{e}}_\omega &=  \dot{e}_\omega + k_R \dot{e}_R \\
&= J^{-1}(- \omega \times J\omega + J(\hat{e}_\omega R_e^T \omega_d - R_e^T \dot{\omega}_d) \\
& \quad + M + \Delta_f) + k_R B e_\omega.
\end{split}
\end{equation*}
A candidate Lyapunov function for the fuselage dynamics is given by $V_2 = V_1 + \frac{1}{2} \tilde{e}_\omega \cdot J \tilde{e}_\omega $. 
$\dot{V}_2$ is given by
\begin{equation} \label{eq:V2dot_temp}
\begin{split}
\dot{V}_2 &= \dot{V}_1 + \tilde{e}_\omega \cdot J \dot{\tilde{e}}_\omega \\
&= -k_R \norm{e_R}^2 + \tilde{e}_\omega \cdot (e_R - \omega \times J\omega + k_R J B e_\omega  + M \\
 &\quad + \Delta_f + J(\hat{e}_\omega R_e^T \omega_d - R_e^T \dot{\omega}_d))  
\end{split}
\end{equation}
Setting $M = M_d$ from \eqref{eq:moment_des} in the above equation would result in
\begin{equation}
\begin{split}
\dot{V}_2 &= -k_R \norm{e_R}^2 - k_\omega \norm{\tilde{e}_\omega}^2 + \tilde{e}_\omega \cdot \left( \Delta_f - \frac{\delta_f^2 \tilde{e}_\omega}{\delta_f \norm{\tilde{e}_\omega} + \epsilon_f} \right) \\
&\leq -k_R \norm{e_R}^2 - k_\omega \norm{\tilde{e}_\omega}^2 + \norm{\tilde{e}_\omega} \delta_f - \frac{\delta_f^2 \norm{\tilde{e}_\omega}^2}{\delta_f \norm{\tilde{e}_\omega} + \epsilon_f} \\
&= -k_R \norm{e_R}^2 - k_\omega \norm{\tilde{e}_\omega}^2 + \epsilon_f \frac{\delta_f \norm{\tilde{e}_\omega}} {\delta_f \norm{\tilde{e}_\omega} + \epsilon_f} \\
&< -k_R \norm{e_R}^2 - k_\omega \norm{\tilde{e}_\omega}^2 + \epsilon_f.
\end{split}
\end{equation}
Adding and subtracting $M_d$ in \eqref{eq:V2dot_temp} would give
\begin{equation}
\begin{split}
\dot{V}_2 &= -k_R \norm{e_R}^2 - k_\omega \norm{\tilde{e}_\omega}^2 + \tilde{e}_\omega \cdot ( \Delta_f + \mu_f + e_M) \\
&< -k_R \norm{e_R}^2 - k_\omega \norm{\tilde{e}_\omega}^2 + \epsilon_f + e_M \cdot \tilde{e}_\omega.
\end{split}
\end{equation}
\\
Augmenting the above Lyapunov function for the fuselage with the quadratic form $\frac{1}{2} e_M\cdot e_M$ gives a candidate Lyapunov function for the complete rotor-fuselage dynamics, $V_3 = V_2 + \frac{1}{2}e_M\cdot e_M$. The derivative of $V_3$ is bounded by
\begin{equation*}
\begin{split}
\dot{V}_3 &= \dot{V}_2 + e_M\cdot \dot{e}_M \\
&< -k_R \norm{e_R}^2 - k_\omega \norm{\tilde{e}_\omega}^2 + \epsilon_f + e_M \cdot \tilde{e}_\omega \\ 
&\quad + e_M \cdot (Ae_M + AM_d - \dot{M}_d - K\omega + KA_\tau \theta).
\end{split}
\end{equation*}
Since $-A_\tau$ is the symmetric part of $A$, the above inequality can be written as
\begin{equation*}
\begin{split}
\dot{V}_3 &< -k_R \norm{e_R}^2 - k_\omega \norm{\tilde{e}_\omega}^2 + \epsilon_f -e_M \cdot A_\tau e_M \\ 
&\quad + e_M \cdot (\tilde{e}_\omega + AM_d - \dot{M}_d - K\omega + KA_\tau \theta).
\end{split}
\end{equation*}
Setting $\theta$ from \eqref{eq:control_input} would make the above inequality 
\begin{equation*}
\dot{V}_3 < -z \cdot W z + \epsilon_f + e_M \cdot ((I - A_\tau \bar{A}_\tau^{-1})\delta_r + A_\tau \bar{A}_\tau^{-1} \mu_r)
\end{equation*}
where $z = (\norm{e_R} \norm{\tilde{e}_\omega} \norm{e_M})$, and $\delta_r$ and $W$ are given in \eqref{eq:temp1} and \eqref{eq:matricesU1U2W} respectively. Now consider the last term of the above inequality
\begin{equation*}
\begin{gathered}
\zeta \triangleq e_M \cdot ((I - A_\tau \bar{A}_\tau^{-1})\delta_r + A_\tau \bar{A}_\tau^{-1} \mu_r), \\
\zeta \leq \max_{\substack{\alpha_m, \alpha_t}} \norm{I-A_\tau\bar{A}_\tau^{-1}} \norm{e_M} \norm{\delta_r} + e_M \cdot A_\tau\bar{A}_\tau^{-1} \mu_r.
\end{gathered}
\end{equation*}
Setting $\mu_r$ from \eqref{eq:temp1} and using the relation in \eqref{eq:A_norm} would result in
\begin{equation*}
\zeta \leq \alpha \norm{e_M} \norm{\delta_r} - \frac{\alpha}{1-\alpha} e_M \cdot A_\tau \bar{A}_\tau^{-1} e_M \frac{\norm{\delta_r}^2}{\norm{\delta_r}\norm{e_M} + \epsilon_r}.
\end{equation*}
Since $A_\tau \bar{A}_\tau^{-1}$ is positive definite and $\frac{\norm{ A_\tau \bar{A}_\tau^{-1}}}{1-\alpha} > 1$,
\begin{equation*}
\begin{split}
\zeta &\leq \alpha \norm{e_M} \norm{\delta_r} - \frac{\alpha \norm{\delta_r}^2 \norm{e_M}^2}{\norm{\delta_r}\norm{e_M} + \epsilon_r} \\
 &= \epsilon_r \frac{\alpha \norm{\delta_r} \norm{e_M}}{\norm{\delta_r}\norm{e_M} + \epsilon_r} < \epsilon_r 
\end{split}
\end{equation*}
Therefore,
\begin{equation}
\dot{V}_3 < -z \cdot W z + \epsilon,
\end{equation}
where $\epsilon = \epsilon_f + \epsilon_r$.

Next the ultimate boundedness for the tracking error dynamics is shown.
$V_3$ is positive definite and quadratic when $\psi(R_e) \leq \xi_2$ for some positive $\xi_2<2$. This is guaranteed when $V_3 \leq \xi_2$.
As a result, $V_3$ satisfies the following inequality in the sublevel set $L_{\xi_2}^- \triangleq \{ (R_e,\tilde{e}_\omega, e_M)\in SO(3) \times \mathbb{R}^3 \times \mathbb{R}^3 | V \leq \xi_2 \}$
\begin{equation}
z\cdot U_1 z \leq V_3 \leq z\cdot U_2 z \\
\end{equation}
or
\begin{equation}
\lambda_{min}(U_1) \norm{z}^2 \leq V_3 \leq \lambda_{max}(U_2) \norm{z}^2
\end{equation}
for positive definite $U_1$ and $U_2$ given by \eqref{eq:matricesU1U2W}.
$\dot{V}_3$ along the solution of error dynamics is guaranteed to be negative definite when
\begin{equation}
-z\cdot W z + \epsilon \leq -\norm{z}^2\lambda_{min}(W) + \epsilon \leq  -V_3\frac{\lambda_{min}(W)}{\lambda_{max}(U_2)} + \epsilon \leq 0 \\
\end{equation} 
or
\begin{equation}
V_3 \geq \left\lbrace \epsilon \frac{\lambda_{max}(U_2)}{\lambda_{min}(W)} \triangleq \xi_1 \right\rbrace
\end{equation}
or in the superlevel set $L_{\xi_1}^+ \triangleq \{ (R_e,\tilde{e}_\omega, e_M)\in SO(3) \times \mathbb{R}^3 \times \mathbb{R}^3 | V \geq \xi_1 \}$. \\
Condition \eqref{eq:lemma1_condt} of lemma 1 is satisfied in the set $\Lambda_{\xi_1}^{\xi_2} \triangleq L_{\xi_1}^- \cap L_{\xi_2}^+ $ and $\xi_1 < \xi_2$ is met by \eqref{eq:epsilon_condt}. Therefore it follows from Lemma 1 that the solutions of the rotor-fuselage error dynamics are uniformly ultimately bounded and the ultimate bound is given by \eqref{eq:ultimate_bound}.

\end{proof}

%%%%%%%%%%%%%%%%%%%%%%%%%%%%%%%%%%%%%%%%%%%%%%%%%%%%%%%%%%%%%%%%%%%%%%%%%%%%%%%%%%%%%%%%%%%%%%%

\begin{remark}
$\epsilon_f$ and $\epsilon_r$ could be independently set based on the uncertainties associated with the fuselage and rotor dynamics. This is an important design flexibility for a helicopter since it allows for adjusting the robustness of the controller for exogenous torque independent of uncertainties in rotor parameters. The exogenous torque depends on the type of mission the helicopter flies (e.g. externally attached payload, cable suspended load), while the rotor parameters remain constant for a given rotor hub and blade properties.
\end{remark}

\begin{remark}
The proposed controller requires flap angle feedback, which in the case of a large scale helicopter is relatively easy to be measured as described in \cite{kufeld1994flight}. Whereas, for a small scale helicopter the instrumentation required for flap angle measurement is challenging because of limited space and the rotor being hingeless in flap. However, an observer for the flap angle can be implemented with the assumption that the remaining states are available. The attitude and angular velocity can be independently estimated using onboard inertial measurement unit as proposed in \cite{mahony2008nonlinear}.
\end{remark}

%%%%%%%%%%%%%%%%%%%%%%%%%%%%%%%%%%%%%%%%%%%%%%%%%%%%%%%%%%%%%%%%%%%%%%%%%%%%%%%%%%%%%%%%%%
%%%%%%%%%%%%%%%%%%%%%%%%%%%%%%%%%%%%%%%%%%%%%%%%%%%%%%%%%%%%%%%%%%%%%%%%%%%%%%%%%%%%%%%%%%

\section{Structure Preserving Robust Controller} \label{sec:str_rob_cont}
In this section we introduce a structure preserving robust attitude tracking controller which achieves almost globally asymptotic stability. The idea here is to preserve the damping term (due to $-K\omega$) inherently present in the rotor dynamics. The backstepping technique used to derive the controller introduced in the previous section necessitated the removal of this term, and introduced artificial damping in the fuselage error dynamics \eqref{eq:moment_des}  through the $-k_{\omega}\tilde{e}_{\omega}$ term. The perfect knowledge of rotor time constant, $\tau_m$, played a crucial role in removal of the damping term $-K\omega$, and therefore its uncertainty resulted in injection of energy into the system through a positive feedback of $\omega$, thus making the system unstable. The present controller avoids this issue by utilizing the inherent damping in the rotor dynamics and doing away with the artificial damping term altogether, thus making it an angular velocity free attitude tracking controller.

The notion of almost global asymptotic stability (AGAS) for mechanical systems defined on non-Euclidean space was introduced by Koditschek in \cite{koditschek1989application}. Due to the topology of $SO(3)$, irrespective of the controller used, it is not possible to make the desired equilibrium globally asymptotically stable using continuous control input as shown in \cite{koditschek1989application,bhat2000topological}. But the region of attraction of the desired equilibrium can be made as large as the entire state space excluding a set of measure zero. 

The proof of AGAS property is outlined in the following three steps \cite{chaturvedi2009asymptotic, bayadi2014almost}.
\begin{enumerate}
\item Define a configuration error function \cite{koditschek1989application, bullo2004geometric} on $SO(3)$ with certain properties. The proportional action of the controller is derived from this function. 

\item Linearize the error dynamics about all equilibria to show that the desired equilibrium is asymptotically stable and the rest are unstable. 

\item Use invariance principle like arguments to show AGAS of the desired equilibrium.
\end{enumerate} 
To show AGAS, the configuration error function should be such that a) it has minimum number of isolated critical points as allowed by the topology of the configuration space (for $SO(3)$ it is 4),
b) the associated gradient vector field be asymptotically stable only on the desired configuration and unstable for the rest.
The modified trace function $\psi_m: SO(3) \to \mathbb{R}$ introduced by Chillingworth \cite{chillingworth1982symmetry} and later used by Koditschek \cite{koditschek1989application} satisfies these properties as given in the following lemma.
\begin{lemma} (Chillingworth \cite{chillingworth1982symmetry}).
If $P \in \mathbb{R}^{3 \times 3}$ is symmetric positive definite with distinct eigenvalues, then the configuration error function
\begin{equation}
\psi_m \triangleq \frac{1}{2} tr(P(I-R_e))
\end{equation}
has exactly 4 critical points.
\end{lemma}
The critical points of $\psi_m$ are $\Theta = \{I, e^{\pi \hat{v}_1}, e^{\pi \hat{v}_2}, e^{\pi \hat{v}_3} \}$, where $v_1, v_2, v_3$ are the eigenvectors of $P$ (p. 553 \cite{bullo2004geometric}). As shown in Eq \eqref{eq:prop_cont}, the proportional control term is obtained from the error corresponding to the modified trace function $\psi_m$ as
\begin{equation}
e_{Rm} = \frac{1}{2} (PR_e - R_e^T P)^\vee.
\end{equation}
With the proportional term defined, the structure preserving controller is given by 
\begin{equation} \label{eq:control_input1}
\theta = (KA_\tau)^{-1}(-AM_d + \dot{M}_d + KR_e^T\omega_d),
\end{equation}
where
\begin{equation} \label{eq:moment_des1}
\begin{gathered}
M_d = -k_R e_{Rm} + \omega \times J \omega
 -J(\hat{e}_\omega R_e^T \omega_d - R_e^T \dot{\omega}_d). \\
\end{gathered}
\end{equation}
By substituting the above controller \eqref{eq:control_input1}, \eqref{eq:moment_des1} in \eqref{eq:rotor-fuse}, with $\Delta_f \equiv 0$, the following error dynamics is obtained
\begin{subequations} \label{eq:err_dym1}
\begin{equation}  \label{eq:err_dym1_kin}
\dot{R}_e = R_e e_\omega
\end{equation}
\begin{equation}  \label{eq:err_dym1_fus}
J\dot{e}_\omega = -k_R e_{Rm} + e_M
\end{equation}
\begin{equation} \label{eq:err_dym1_rot}
\dot{e}_M = A e_M - K e_\omega
\end{equation} 
\end{subequations}
The above error dynamics has 4 equilibria given by $\chi = \{ (R_{eq} \in \Theta, e_{\omega} = 0, e_M = 0)  \}$, with $(I,0,0)$ being the desired equilibrium. To show the local stability properties, the error dynamics is linearized about each equilibrium. In case of a general smooth manifold, the vector field is mapped to the coordinate chart and then linearized \cite{bayadi2014almost}. Since the configuration manifold in our case is a Lie group, we use the simplified approach followed by Chaturvedi et al \cite{chaturvedi2009asymptotic}. Using the exponential map, a small perturbation to the given equilibrium $(R_{eq},0,0)$ is represented by $(R(\epsilon),\epsilon \bar{e}_{\omega}, \epsilon \bar{e}_M)$ with $R(\epsilon) = R_{eq}e^{\epsilon \hat{\bar{\eta}}}$, for a small $\epsilon \in \mathbb{R}$ and linearization variables $\bar{\eta}, \bar{e}_{\omega}, \bar{e}_M \in \mathbb{R}^3$. The linearization is done by substituting the perturbed states to the error dynamics \eqref{eq:err_dym1} and differentiating w.r.t $\epsilon$ at $\epsilon = 0$. 
\begin{lemma}
The linearization of the error dynamics \eqref{eq:err_dym1} is given by 
\begin{equation} \label{eq:linerr_dym1}
\frac{d}{dt} \begin{bmatrix}
\bar{\eta} \\
\bar{e}_\omega \\
\bar{e}_M 
\end{bmatrix} = 
\underbrace{
\begin{bmatrix}
\pmb{0} & \pmb{I} & \pmb{0}  \\
-J^{-1} k_R B(R_{eq}) & \pmb{0} & J^{-1} \\
\pmb{0} & -K & A  
\end{bmatrix}}_{\textstyle S(R_{eq})}
\begin{bmatrix}
\bar{\eta} \\
\bar{e}_\omega \\
\bar{e}_M 
\end{bmatrix}
\end{equation}
where $B(R_{eq}) = -\frac{1}{2}\sum_{i=1}^3 \hat{e}_i P R_{eq} \hat{e}_i$, and $\pmb{I}$ is the $3\times 3$ identity matrix. 
\end{lemma}
\begin{proof}
See Appendix \ref{ap:lin}.
\end{proof}
Note that, the stability matrix $S(R_{eq})$ depends only on the critical points of $\psi_m$ in $\Theta$. The eigenvalues of $S(R_{eq})$ reveal that for all $k_R > 0$, each equilibrium is hyperbolic and the desired equilibrium is  asymptotically stable while the rest are unstable. The center manifold theorem (Theorem 3.2.1, \cite{guckenheimer2013nonlinear}) states that the tangent space to the unstable and stable manifold of the equilibrium coincide with the unstable and stable eigen spaces of the linearized system, respectively. Further, these manifolds are invariant with respect to the flow of the vector field. Therefore, the stable manifolds of the unstable equilibria are lower dimensional than the total space $SO(3)\times \mathbb{R}^3 \times \mathbb{R}^3$, thus making them a set of measure zero. Now we present the main result of the structure preserving controller.

\begin{theorem}
For $k_R > 0$, the desired equilibrium $(I,0,0)$ of the error dynamics \eqref{eq:err_dym1} is almost globally asymptotically stable.
\end{theorem}

\begin{proof}

Consider the following candidate Lyapunov function $V = k_R\psi_m(R_e) + \frac{1}{2} e_\omega \cdot J e_\omega + \frac{1}{2}e_M \cdot K^{-1} e_M$. Its derivative along the error dynamics vector field is given by
\begin{equation}
\begin{split}
\dot{V} = k_R e_{Rm} \cdot e_\omega + e_\omega \cdot (e_M - k_R e_{Rm}) \\
 + e_M \cdot K^{-1}(Ae_M - Ke_\omega) \\
 = e_M \cdot K^{-1} A e_M \leq 0.
\end{split}
\end{equation}
The negative definitiveness of $K^{-1}A$ follows from the fact that the upper diagonal 2x2 block of $A$ is a rotation transformation by an obtuse angle and scaling, and $K>0$ is diagonal with first two elements being equal. Since $V$ is continuous and bounded from below and $\dot{V} \leq 0$, the positive limit set of all the trajectories is characterized by $\dot{V} \equiv 0$. The following sequence of arguments show that the union of such limit sets is exactly the set of 4 equilibrium points, $\chi$.
\begin{equation}
\begin{split}
\dot{V} \equiv 0 \implies e_M \equiv 0 \implies \dot{e}_M \equiv 0 \implies e_\omega \equiv 0 \implies \\
 \dot{e}_\omega \equiv 0 \implies e_{Rm} \equiv 0.
\end{split}
\end{equation}
From the above identity and from the arguments made previously, all the initial conditions that start outside the stable manifold of the unstable equilibria converge to the desired equilibrium asymptotically. Since the stable manifolds of the unstable equilibria are of measure zero, the desired equilibrium is almost globally asymptotically stable.
\end{proof}

Now we show the robustness of the above controller using input to state stability (ISS) arguments. The following lemma helps prove the robustness of the proposed controller.
\begin{lemma} \label{lem:rob} (Theorem 7.4, \cite{marquez2003nonlinear}).
Consider the system $\dot{x} = f(x,u)$. Assume that the origin is an asymptotically stable equilibrium point for the autonomous system $\dot{x} = f(x,0)$, and that the function f(x,u) is continuously differentiable. Under these conditions $\dot{x} = f(x,u)$ is locally input to state stable.
\end{lemma}
An uncertainty in rotor time constant modifies the rotor error dynamics \eqref{eq:err_dym1_rot} as
\begin{equation} \label{eq:err_dym2_rotor}
\dot{e}_M = A e_M - K e_\omega + \Delta_r,
\end{equation}
where $\Delta_r = (A_{\tau} \bar{A}_{\tau}^{-1} - I)(-AM_d + \dot{M}_d + KR_e^T\omega_d)$. Note that $\Delta_r$ is bounded as $\omega_d$ and $e_R$ are bounded. Since the error dynamics without the disturbance $\Delta_r$ is shown to be asymptotically stable, it follows form Lemma \ref{lem:rob} and uncertainty bound \eqref{eq:A_norm} that a small uncertainty in rotor time constant would result in a small tracking error.

\begin{remark}
The error dynamics \eqref{eq:err_dym1} has the same structure as that of the helicopter dynamics \eqref{eq:rotor-fuse} except for the gyroscopic term $\omega \times J \omega$ and the feedback term $-k_R e_{Rm}$. The gyroscopic term, due to its energy conserving nature, does not change the damping characteristic of the system. As a result the closed loop system retains the same damping characteristic of the actual helicopter as depicted in Fig. \ref{fig:damping}. This is an advantage as the inherent damping present in the system is very large and relieves the controller from introducing artificial damping which makes it angular velocity feedback free. This distinguishes the present controller from pure rigid body control.
\end{remark}
\begin{remark}
The controller proposed in this section eliminates the major drawback of the robust controller in Sec \ref{sec:rob_cont} as it does not require flap angle and angular velocity for feedback. The only feedback term is $-k_R e_{Rm}$, which only depends on the easily measurable current attitude $R(t)$. This is particularly important for severe aerobatic maneuvers as in such cases the angular velocity about the minimum inertia axis ($X_b$) is corrupted with high rotor vibration noise.
\end{remark}

Due to its ease of implementation, the performance of the above controller is validated with extensive experiments and is provided in Sec \ref{sec:Exp}.

%%%%%%%%%%%%%%%%%%%%%%%%%%%%%%%%%%%%%%%%%%%%%%%%%%%%%%%%%%%%%%%%%%%%%%%%%%%%%%%%%%%%%%%%%%
%%%%%%%%%%%%%%%%%%%%%%%%%%%%%%%%%%%%%%%%%%%%%%%%%%%%%%%%%%%%%%%%%%%%%%%%%%%%%%%%%%%%%%%%%%

\section{Simulation Results} \label{sec:sim}

This section presents the comparative performance of the nominal controller ($\mu_f = \mu_r = 0$) and the robust controller presented in Sec \ref{sec:rob_cont}. The tracking controller was simulated for a 10 kg class model helicopter whose parameters are given in Table \ref{tab:heli_params}. To study the effect of individual uncertainties and the efficacy of each robustification term ($\mu_f,\mu_r$), independent simulations were carried out for structured and unstructured uncertainty. Next, the uncertainties were applied simultaneously to study their combined effect on the performance of the proposed controller.

For the purpose of uniformity in results and easy comparison, the reference trajectories and the initial conditions were chosen to be identical throughout all simulations. The reference trajectory for tracking was designed such that it requires large control input and is sufficiently fast enough to be termed aggressive. A reasonable such candidate is a sinusoidal roll angle reference with an amplitude of 20 degree and a frequency of 1 Hertz. This maneuver requires about 8 degree cyclic input, which is almost 80 percent of the maximum allowed input for aerobatic helicopters of this class. A random large initial attitude error of 80 deg in pitch angle and 90 deg/s of pitch-rate was prescribed for all simulations. The controller parameters used for simulation are given in Table \ref{tab:simulation_params}.

\begin{table} [tbp!]
\renewcommand{\arraystretch}{1.3}
\caption{Helicopter parameters}
\label{tab:heli_params}
\centering
\begin{tabular}{l l l} 
 \hline
 Parameter & Description & Values  \\
 \hline
 $[J_{xx} J_{yy} J_{zz}]$ & Moment of inertia & [0.095 0.397 0.303] $kg$-$m^2$  \\ 

 $\tau_m$  & Rotor time constant & 0.06 $s$  \\

 $k_\beta$  & Rotor spring constant & 129.09 $N$-$m$ \\

 $I_\beta$  & Blade inertia & 0.0327 $kg$-$m^2$ \\
 
 $\Omega$ & Rotor speed & 157.07 $rad/s$\\
 
 $h$ & Hub distance from c.g & 0.174 $m$  \\ [1ex] 
 \hline
\end{tabular}
\end{table}

\begin{table} [tbp!]
\renewcommand{\arraystretch}{1.3}
\caption{Simulation parameters}
\label{tab:simulation_params}
\centering
\begin{tabular}{l l l} 
 \hline
 Parameter & Description & Values  \\
 \hline
 $k_R$    & BRC P gain   & 2.8  \\ 

 $k_\omega$ & BRC D gain  & 2.5  \\

 $\epsilon_f$ & Fuselage error bound & 0.1  \\

 $\epsilon_M$ & Rotor error bound & 0.1  \\

 $\delta_f$   & Max disturbance torque & 5 $N$-$m$ \\

 $A_d$        & Lumped disturbance torque amplitude & 5 $N$-$m$ \\

 $\Omega_d$   & Lumped disturbance torque frequency & 1.5$\pi$ $rad/s$ \\  
 \hline
\end{tabular}
\end{table}

\begin{figure*} [!t] 
\centering
\subfloat[Nominal controller tracking response \label{fig:str_fail}]{\includegraphics[width=0.5\linewidth,height=10cm]{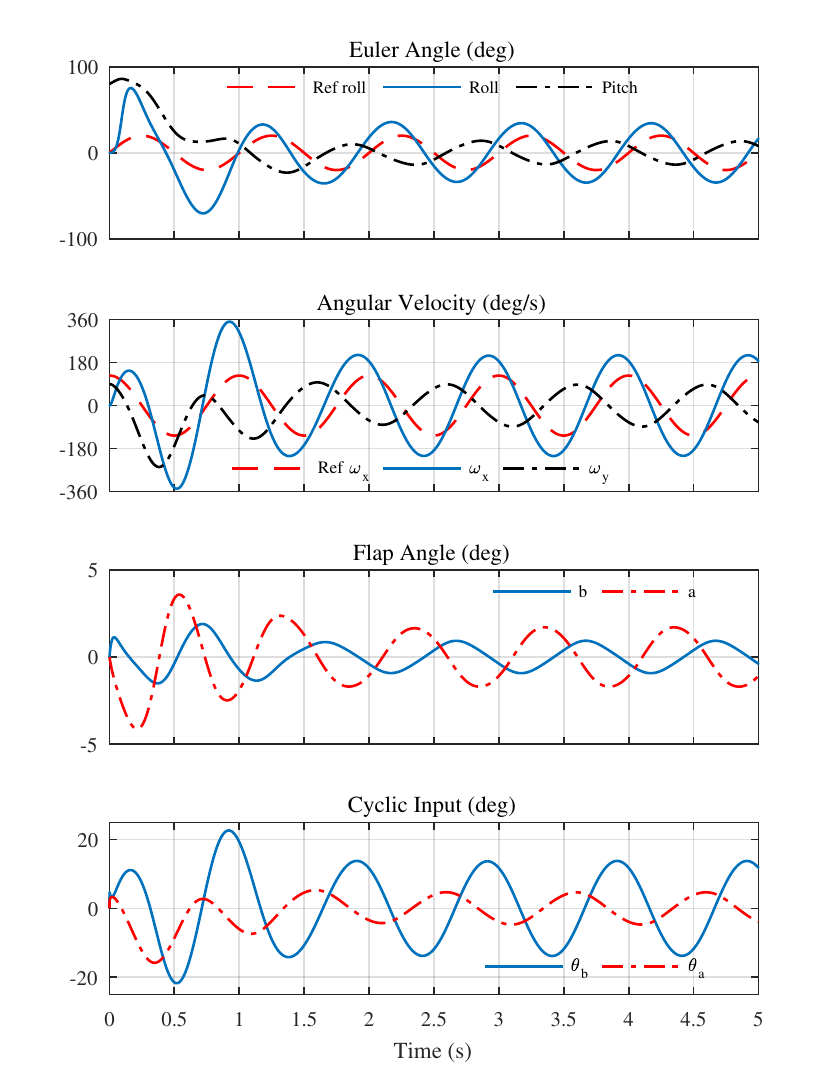}}%
\subfloat[Robust controller tracking response \label{fig:str_good}]{\includegraphics[width=0.5\linewidth,height=10cm]{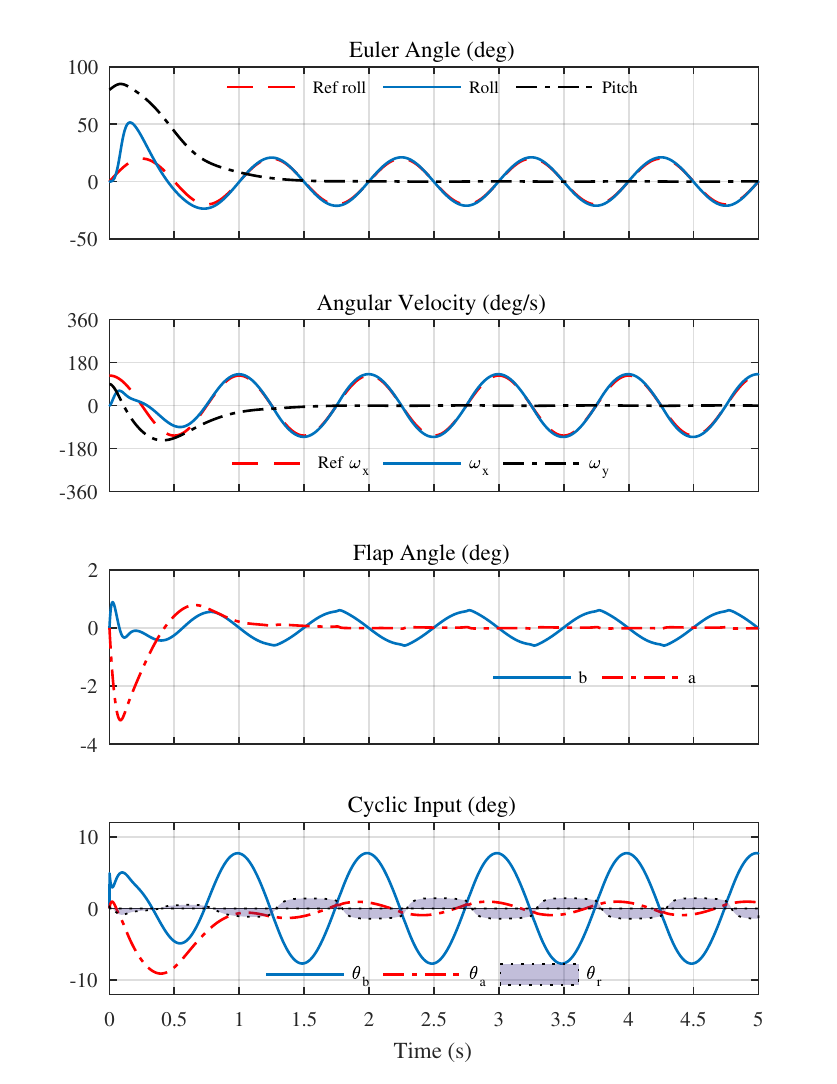}}
\caption{Comparison of Robust and Nominal controller response for structured uncertainty}
\end{figure*}

\begin{figure*} [!t] 
\centering
\subfloat[Nominal controller tracking response \label{fig:unstr_fail}]{\includegraphics[width=0.5\linewidth,height=11cm]{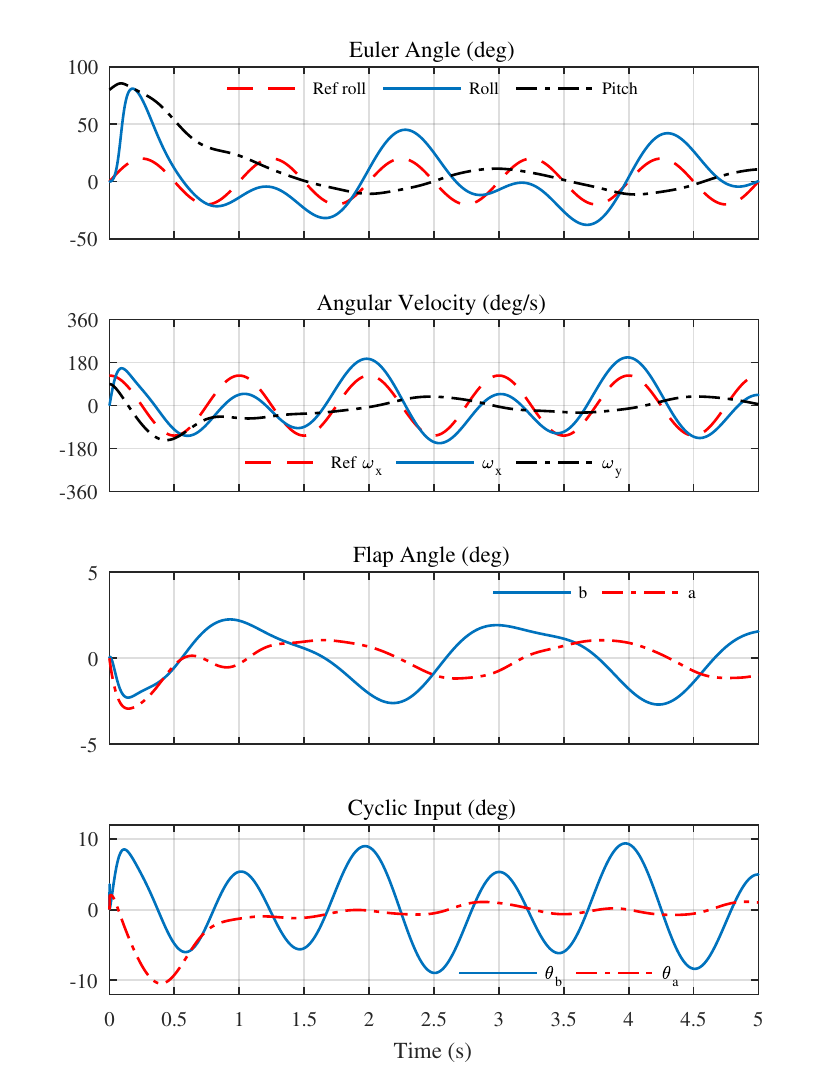}}%
\subfloat[Robust controller tracking response \label{fig:unstr_good}]{\includegraphics[width=0.5\linewidth,height=11cm]{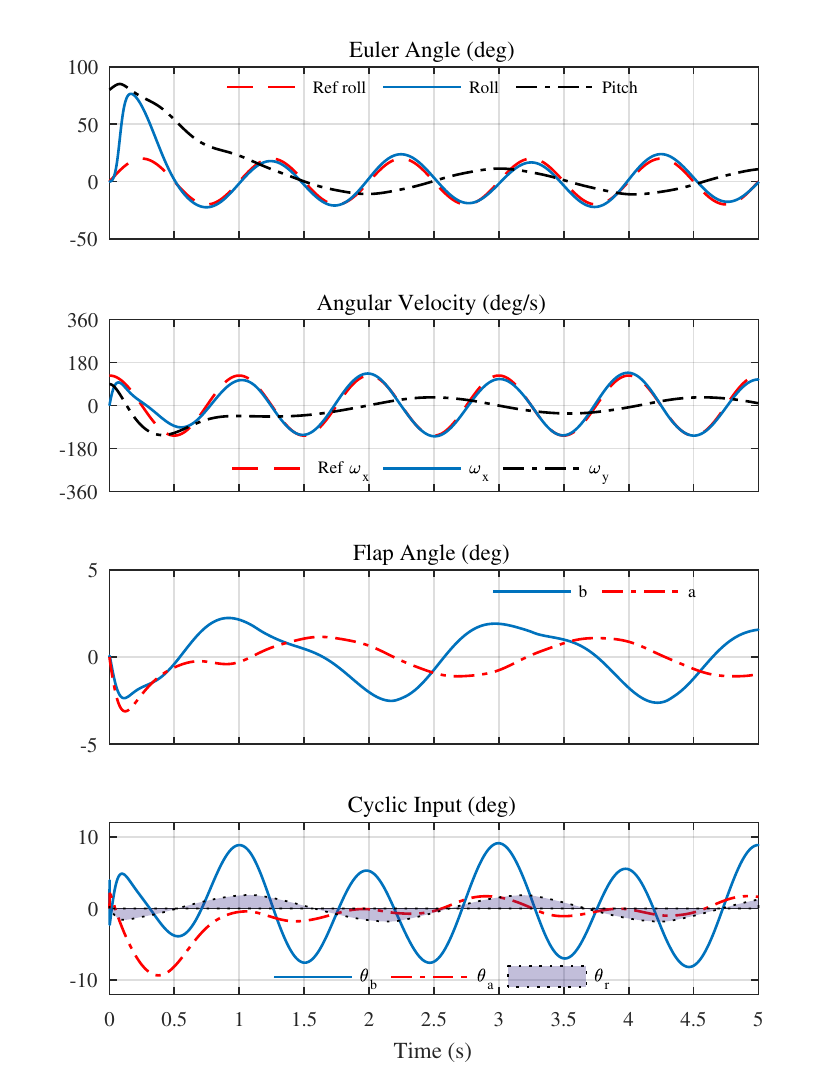}}
\caption{Comparison of Robust and Nominal controller response for unstructured uncertainty}
\end{figure*}

\begin{figure*} [!t] 
\centering
\subfloat[Nominal controller tracking response \label{fig:str_unstr_fail}]{\includegraphics[width=0.5\linewidth,height=11cm]{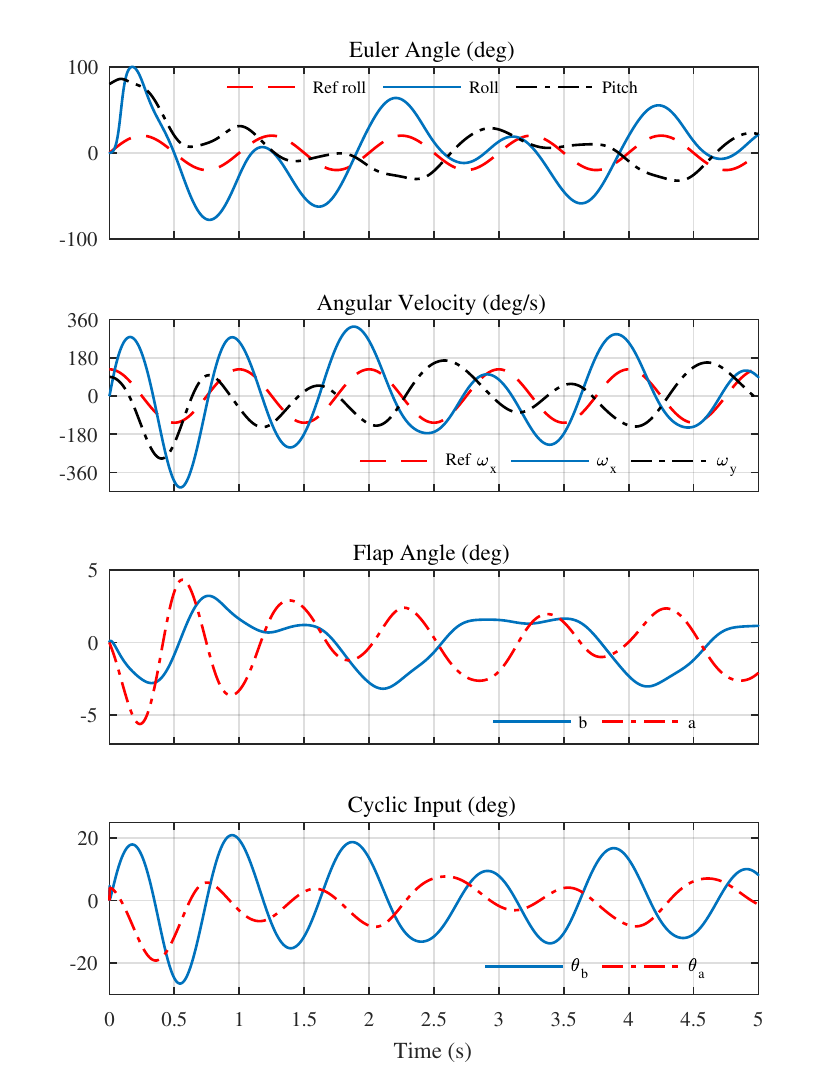}}%
\subfloat[Robust controller tracking response \label{fig:str_unstr_good}]{\includegraphics[width=0.5\linewidth,height=11cm]{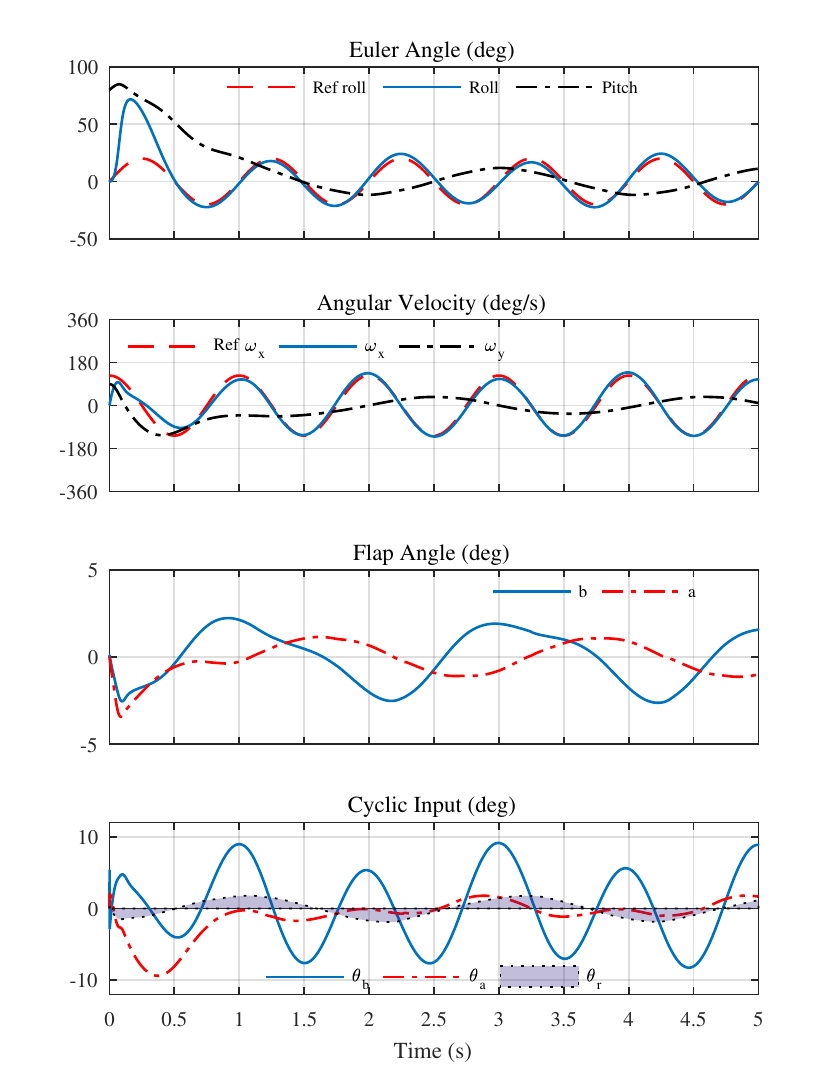}}
\caption{Comparison of Robust and Nominal controller response for combined structured and  unstructured uncertainty}
\end{figure*}

In the case of structured uncertainty, the main rotor time constant used for controller implementation was assumed to be 30 percent more than the simulation model value. This increment would result in a positive feedback of angular velocity thereby injecting energy into the system. As a result, there is a significant tracking error in angle and velocity with the nominal controller as shown in Fig. \ref{fig:str_fail}. On the other hand, the robust controller tracks almost perfectly as is evident from Fig. \ref{fig:str_good}. Moreover, the nominal controller demanded a peak cyclic input of 13.6 degree, which is significantly more than the maximum permissible value of 10 degree. Contrarily, the robust controller demand was well within this bound and the contribution of robust compensation, denoted $\theta_r$, is smooth as shown in Fig. \ref{fig:str_good}.

Unstructured disturbance represents exogenous torque acting on the fuselage primarily due to external payload. A typical case of a swinging under-slung load is simulated here. For the helicopter considered here, a maximum under-slung load equivalent to maximum payload capacity of 3 kg is considered. This results in a maximum torque of 5 N-m for a swing amplitude of 60 degree and any string length. The unknown torque was introduced by adding a lumped disturbance of the form $\Delta_f(t) = A_d[cos(\Omega_d t),0,0]$ to the fuselage. In this case, the performance of the controller with and without the robustification term $\mu_f$ can be compared from Figs. \ref{fig:unstr_fail} and \ref{fig:unstr_good}. It is evident that the robust controller is very effective at nullifying the disturbance torque with modest control requirement.

Figure \ref{fig:str_unstr_good} shows the performance of the proposed robust controller with the combined effect of structured and unstructured disturbances. The nominal controller is completely incapable of tracking the reference command and applies control input much larger than the permissible limit of 10 deg. On the other hand, the robust controller tracks the reference command close enough for all practical purposes with the control input within 10 deg limit. Note that the contribution of the robust compensator is smooth in all the cases considered.

%clearpage
%%%%%%%%%%%%%%%%%%%%%%%%%%%%%%%%%%%%%%%%%%%%%%%%%%%%%%%%%%%%%%%%%%%%%%%%%%%%%%%%%%%%%%
%%%%%%%%%%%%%%%%%%%%%%%%%%%%%%%%%%%%%%%%%%%%%%%%%%%%%%%%%%%%%%%%%%%%%%%%%%%%%%%%%%%%%%
%%%%%%%%%%%%%%%%%%%%%%%%%%%%%%%%%%%%%%%%%%%%%%%%%%%%%%%%%%%%%%%%%%%%%%%%%%%%%%%%%%%%%%

\section{Experimental Results} \label{sec:Exp}
This section presents the experimental validation of the structure preserving controller introduced in Section \ref{sec:str_rob_cont}. The efficacy of the controller is demonstrated by performing single axis flip maneuvers about the roll and pitch axes on a small scale aerobatic helicopter. As the primarily interest in this paper is the coupling between the main rotor and fuselage dynamics, the desired trajectories are specified about roll ($X_b$) and pitch ($Y_b$) axes. Although the maneuvers could be performed about any axes, the specific choice was made to facilitate easy recovery after the maneuver.

\subsection{Desired Trajectory}
The flip maneuvers should meet two conflicting requirements: a) respect the actuator and state constraints of the helicopter, b) aggressive/fast enough for testing the proposed controller and prevent excessive height loss. To this end, the desired rotation trajectories are obtained as a solution to the following optimal control problem:

\begin{equation*}
\underset{u}{\text{min}} 
 \int_{0}^{T_f} u^2 dt \\
\end{equation*}
\begin{align} \label{eq:optimalCP}
%\begin{aligned}
\begin{rcases}
\text{such that}
&  \omega = \dot{\phi}v, \\
&  J\dot{\omega} + \omega\times J\omega = M, \\
&  \dot{M} = AM - K\omega + KA_\tau \theta, \\
&  \dot{\theta} = u  \\
&  \norm{u} \leq u_{max}, \quad \norm{\theta} \leq \theta_{max}
\end{rcases}
\quad \forall t \in [0,T_f]
%\end{aligned}
\end{align}
with boundary conditions
\begin{gather*}
\phi(0) = 0, \quad \phi(T_f) = \phi_{tr}, \quad \omega(0) = 0, \quad \omega(T_f) = 0, \\
\quad M(0) = M_{tr_1}, \quad M(T_f) = M_{tr_2}, \\
 \theta(0) = \theta_{tr_1}, \quad \theta(T_f) = \theta_{tr_2}
\end{gather*}
where $(\cdot)_{tr}$ is the corresponding trim value for hover,  $u$ is the rate of change of blade pitch input, and the attitude is described by the axis-angle representation $(\phi,v)$ 
\begin{equation*}
R_d(\phi,v) = I + \sin\phi\hat{v} + (1-\cos\phi)\hat{v}^2.
\end{equation*}
Since a single axis rotation is considered, $v$ is a constant vector and the optimal attitude trajectory is obtained as time parametrized angle, $\phi(t)$. The proposed optimal control problem allows to include servo speed limit through $u_{max}$, and ensure that the cyclic control input to the helicopter at the beginning and end of the maneuver is its trim value. The maximum cyclic input, $\theta_{max}$, was set to 9.8 deg, which corresponds to a steady state angular velocity of 170 deg/s for the given helicopter. A direct collocation method \cite{betts2010practical} was used to obtain the optimal solution in Matlab.

\subsection{Implementation Details}
The experiments were carried out on Align Trex 700 electric helicopter. It has a rotor diameter of 1.5 meter and weighs 6 kg. The rotor rpm was set to 1500, although a higher rpm would make the vehicle more capable of aerobatic flight. It was instrumented with Pixhawk autopilot board running a modified version of PX4 open source code. The board is equipped with a 3-axis accelerometer,  3-axis gyro and a 3-axis magnetometer together constituting the attitude heading reference system (AHRS). The stock code has an implementation of quaternion based attitude estimator which fuses the data from the AHRS using a complementary filter. The controller was implemented as a separate module and runs at 250 Hz. Due to limited storage capacity of the autopilot ROM, the optimal desired attitude trajectory was approximated using piecewise polynomials of degree 7 and 1, the coefficients of which were stored onboard. Although, the SPR controller gain $k_R$ is a scaler in the theorem, for practical implementation it can be chosen to be a positive diagonal matrix to get the desired handling quality about different axes. It can be chosen based on the desired frequency characteristic of the linearized stability matrix $S(I)$ \eqref{eq:linerr_dym1}.

\subsection{Experimental Results}
For the purpose of validation, two flips of 180 and 360 degree each about roll and pitch axes were performed. As the vehicle is capable of generating negative thrust, it can hover in the flipped upside down configuration, thus making this a safe maneuver. In order to reduce the height loss during flip, the collective (thrust) input was specified as $\theta_0 = \theta_{0(hov)}Z_b\cdot Z_e$ (see Fig \ref{fig:heli_model}). The trim values of the state and control input were obtained for the upright configuration experimentally and determined for the inverted configuration by symmetry. The total flip duration $T_f$ was specified based on actuator saturation and safety limits of the vehicle. For the 180 deg flip, it was chosen to be 1.2 second and 2.3 second for 360 degree flip. Although multiple flips could be performed with the proposed controller, the total maneuver angle is limited by the associated height loss.

The experimental data collected from the maneuvers are shown in Fig. \ref{fig:exp_plot}. The roll and pitch maneuvers are represented in the plot respectively by 321 and 312 Euler angles to avoid singularity. It is observed that the commanded angles for both the axes are tracked almost perfectly. The maximum cyclic input for roll ($\theta_{b}$) and pitch ($\theta_a$) were limited to 10.5 deg. It is evident from Fig. \ref{fig:exp_plot} that the actuators are almost saturated throughout the maneuvers. The small deviation from perfect tracking can be attributed mainly to the following phenomena not captured by the simple model of helicopter: a) Unsteady aerodynamics b) Higher order flap/lead-lag dynamics c) Servo dynamics. Regardless of the above artifacts, the ability of the controller to track the aggressive flip trajectory demonstrates the robustness of the structure preserving controller and the validity of the minimal helicopter model given by \eqref{eq:rotor-fuse} at the limit of performance of the vehicle.

\begin{figure*} [!t] 
\centering
\subfloat[Roll flip 180 deg \label{fig:exp_r180}]{\includegraphics[height=11cm, width=8cm]{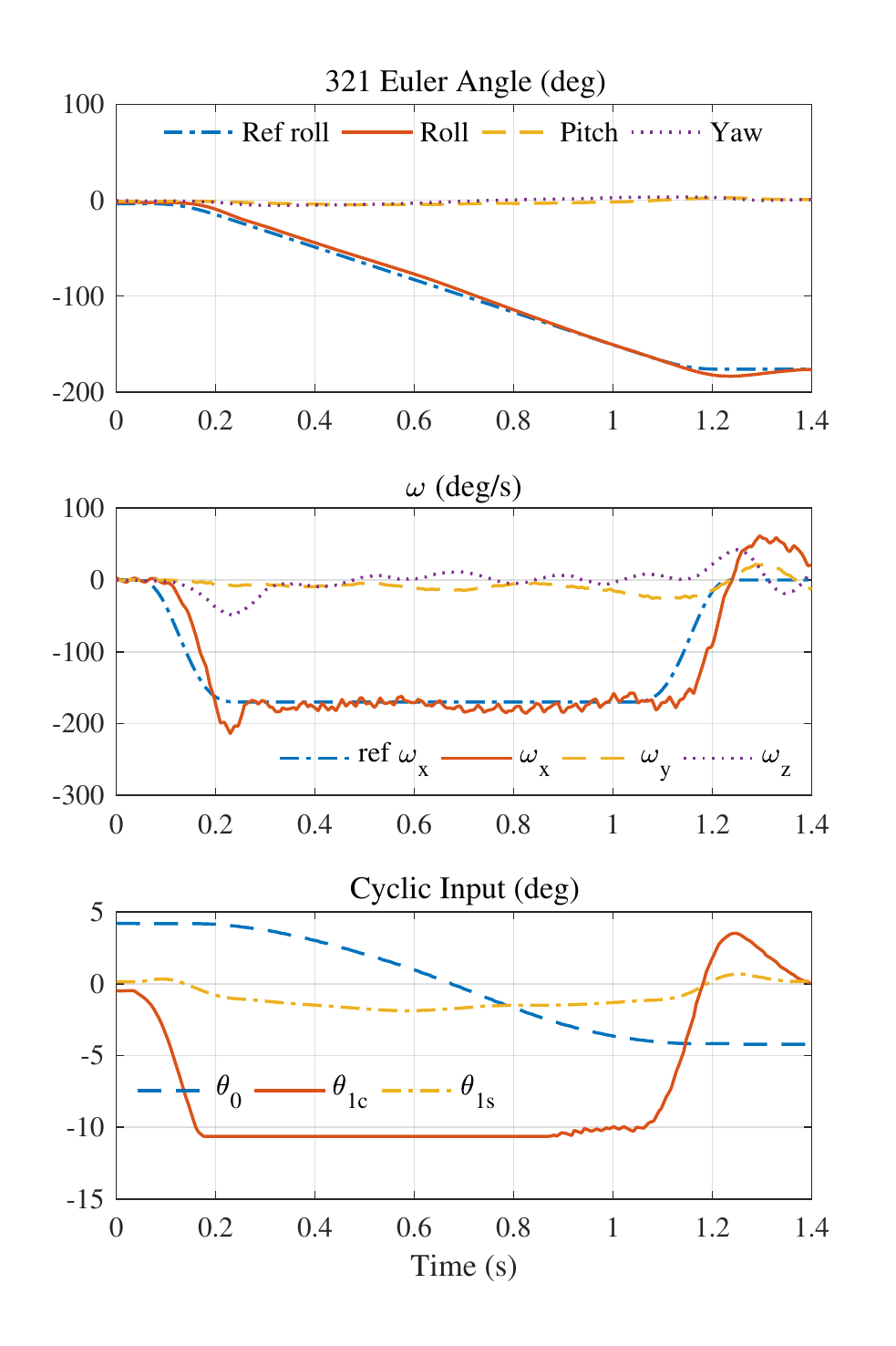}}%
\subfloat[Pitch flip 180 deg \label{fig:exp_p180}]{\includegraphics[height=11cm, width=8cm]{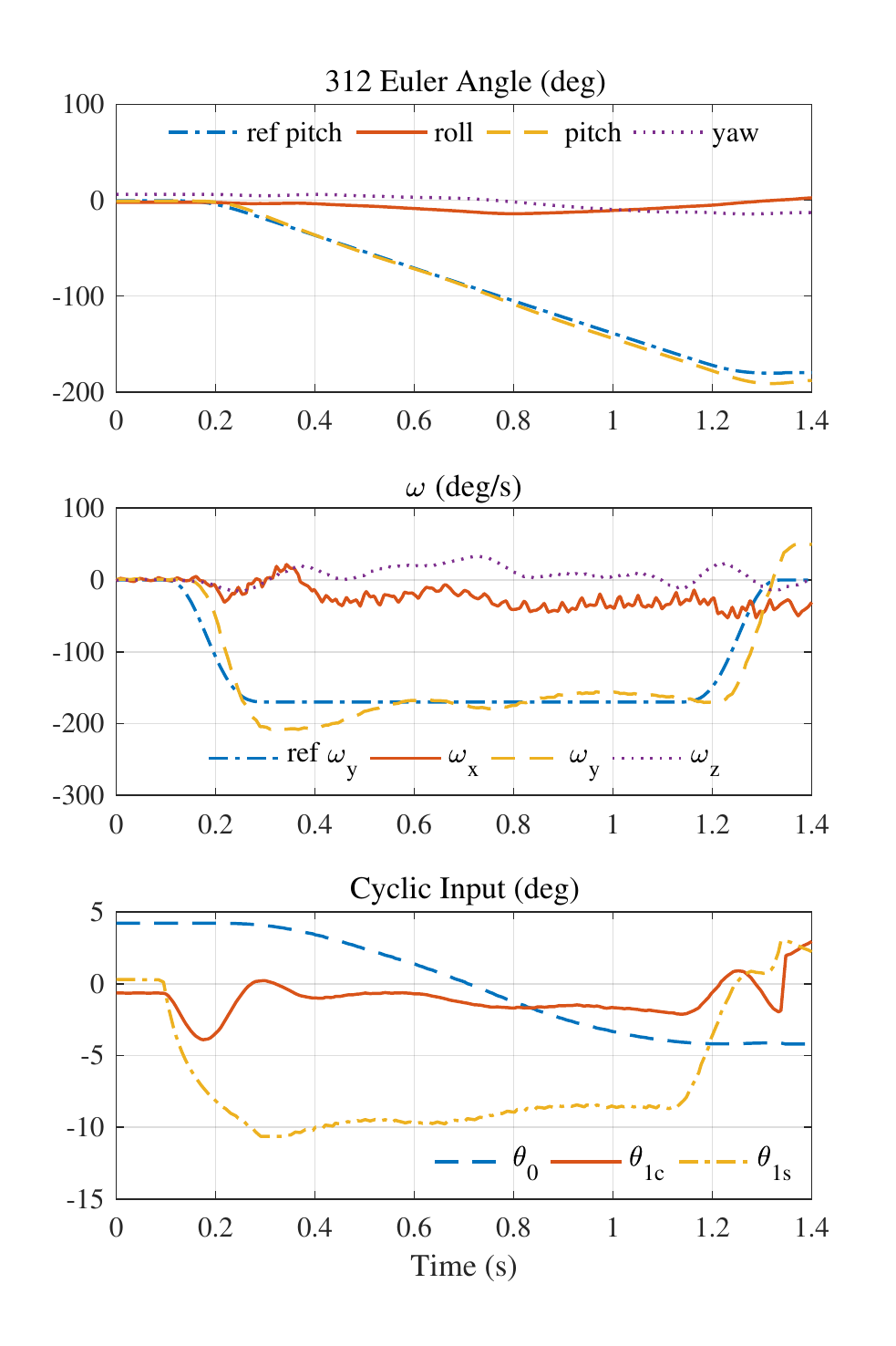}}\\
\subfloat[Roll flip 360 deg \label{fig:exp_r360}]{\includegraphics[height=11cm, width=8cm]{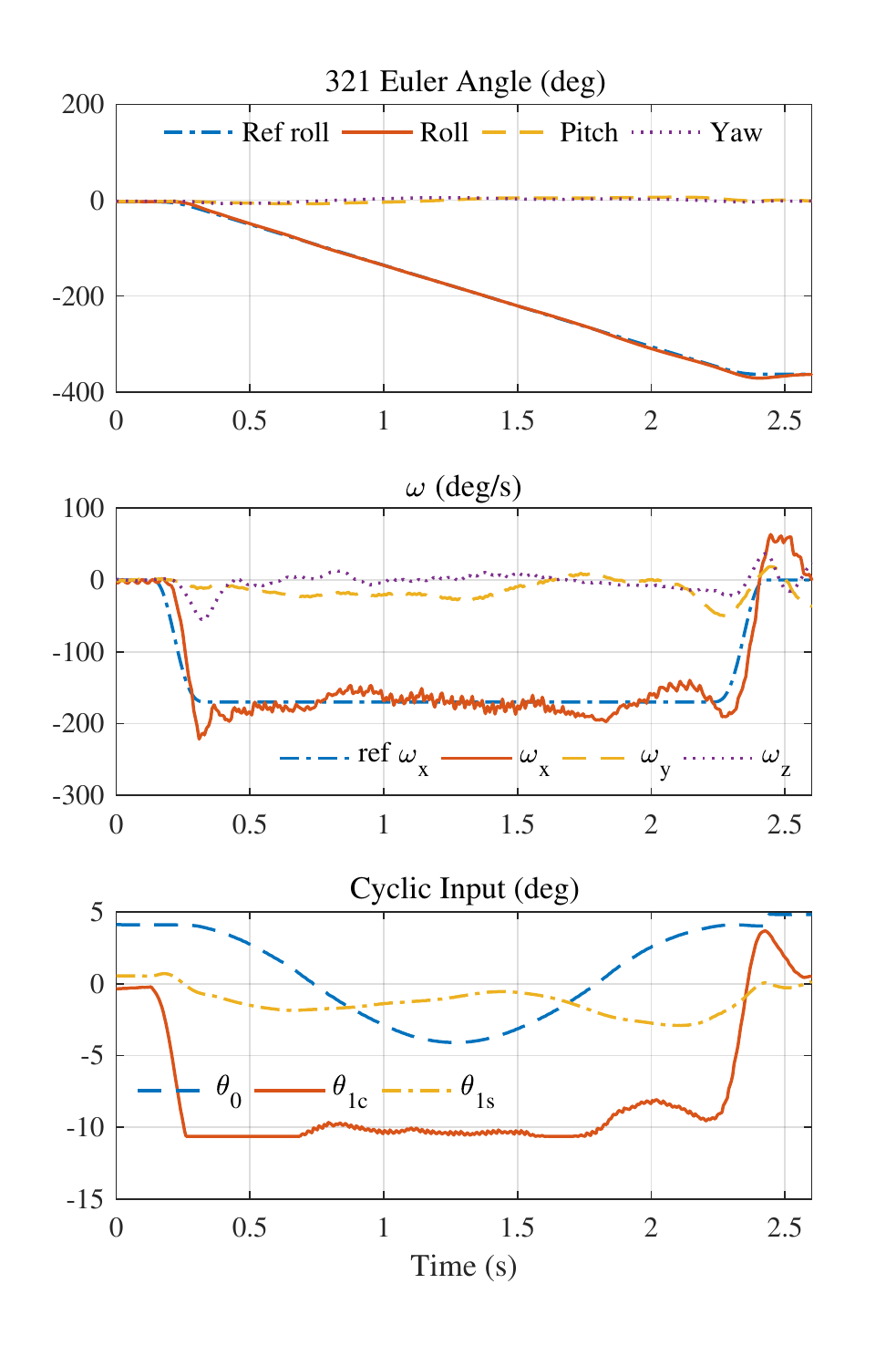}}%
\subfloat[Pitch flip 360 deg \label{fig:exp_p360}]{\includegraphics[height=11cm, width=8cm]{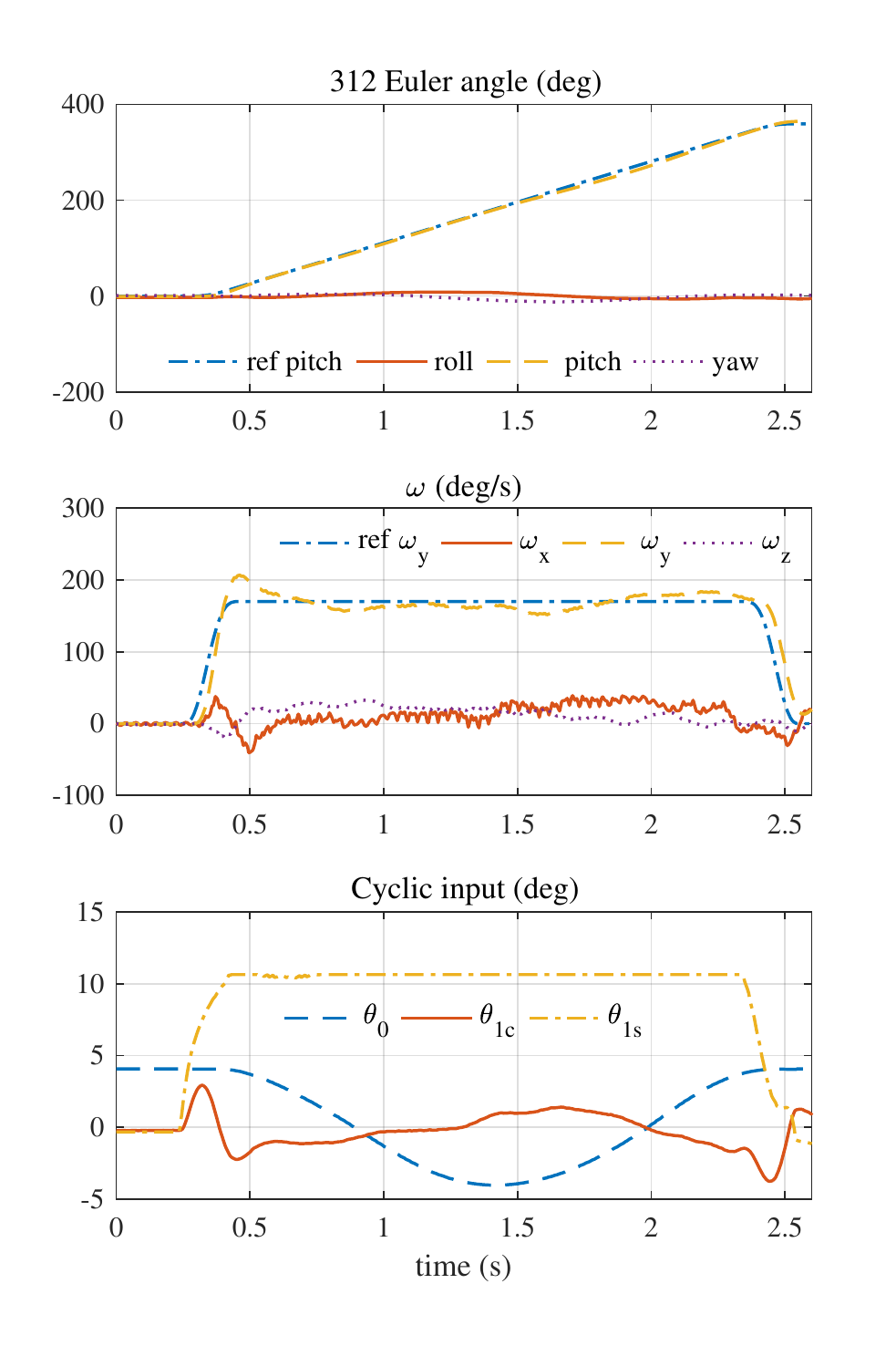}}
\caption{Experimental validation of the structure preserving controller by performing roll and pitch flip.}
\label{fig:exp_plot}
\end{figure*}

\begin{figure*} [!t] 
\centering
\subfloat[Roll flip 180 deg \label{fig:r180_pic}]{\includegraphics[height=11cm, width=8cm]{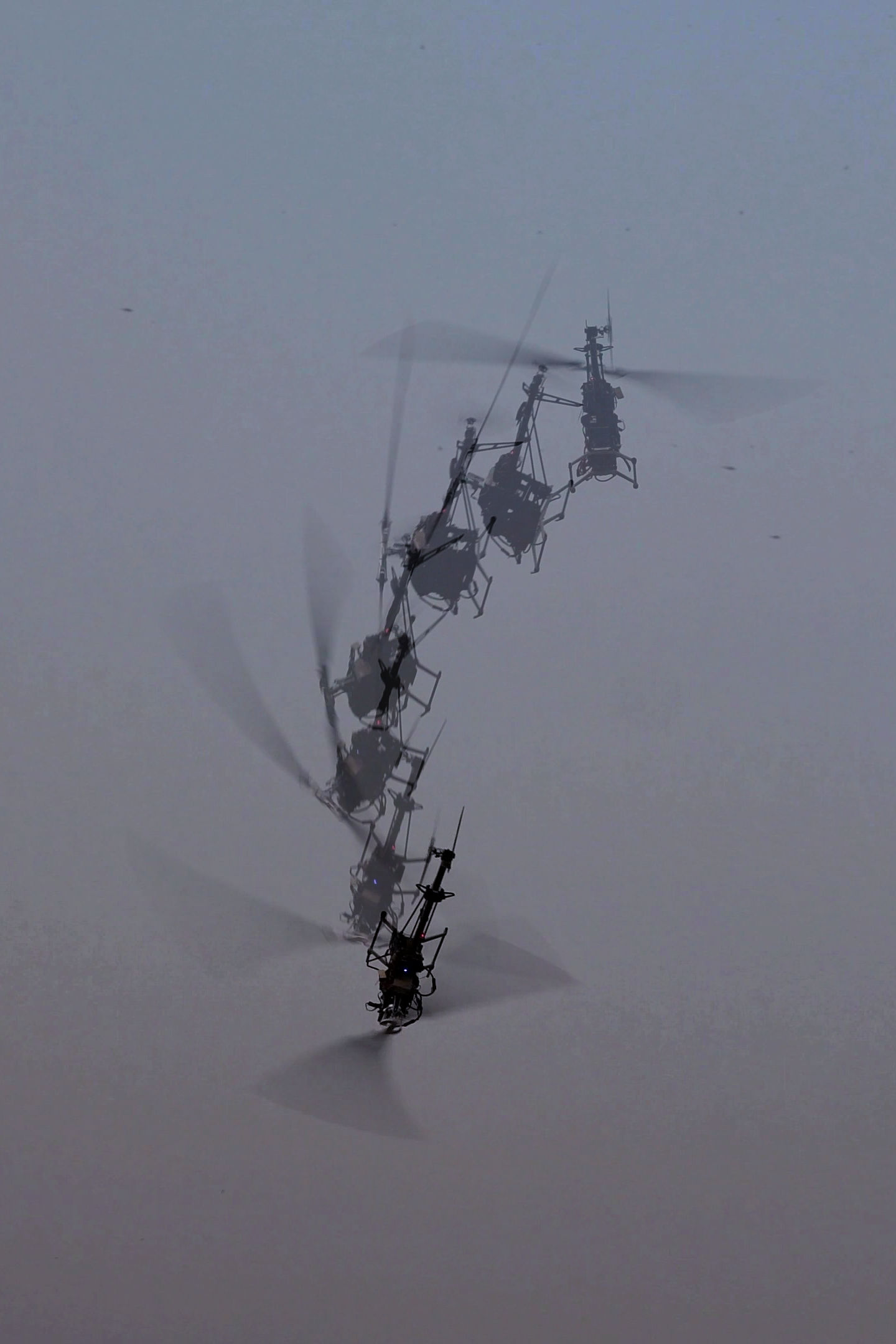}} \qquad
\subfloat[Pitch flip 180 deg \label{fig:p180_pic}]{\includegraphics[height=11cm, width=8cm]{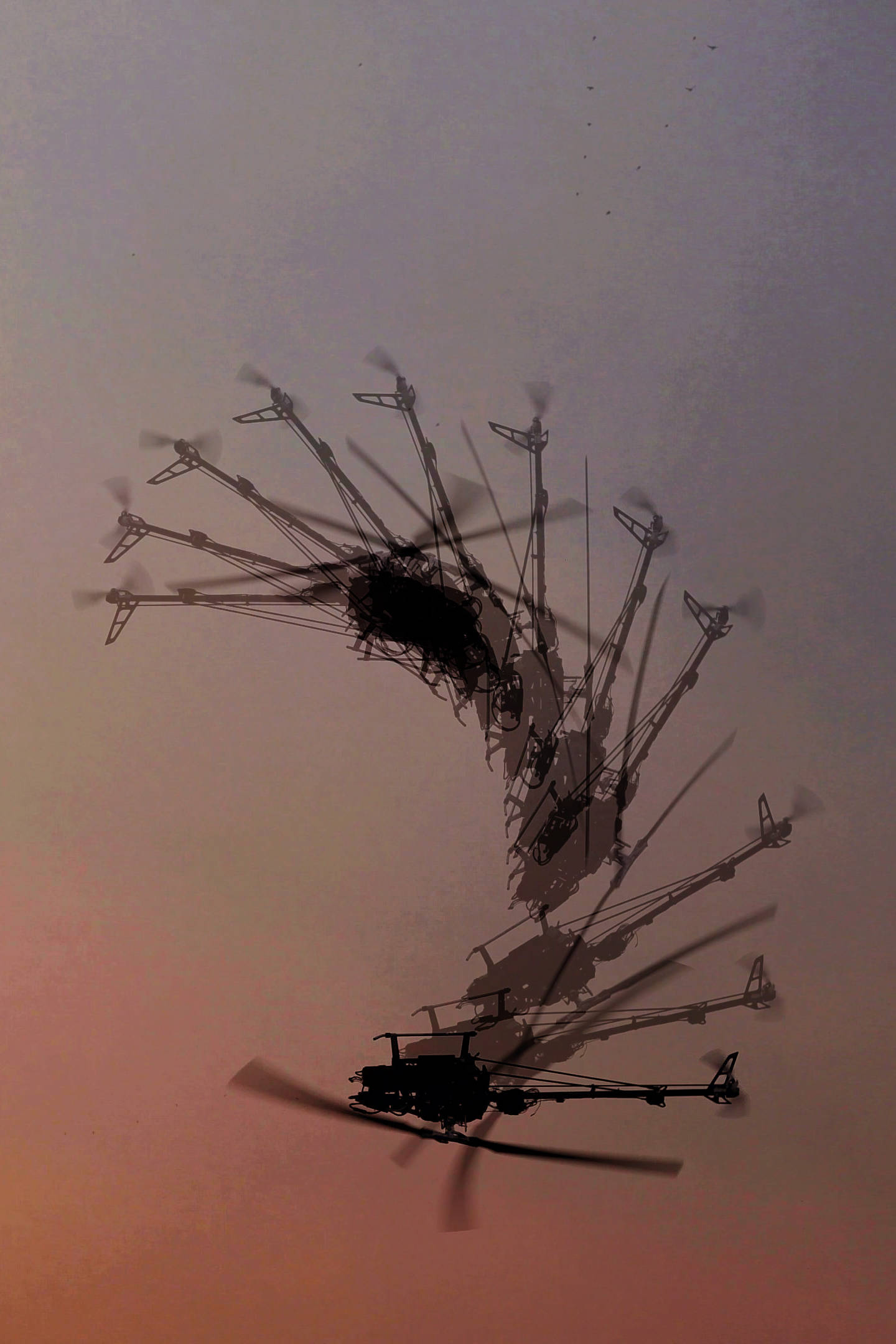}}\\
\caption{Instants during the flip maneuver. Video link: \url{https://youtu.be/1zz71W__RNA}}
\label{fig:flip_pic}
\end{figure*}

%%%%%%%%%%%%%%%%%%%%%%%%%%%%%%%%%%%%%%%%%%%%%%%%%%%%%%%%%%%%%%%%%%%%%%%%%%%%%%%%%%%%%%
%%%%%%%%%%%%%%%%%%%%%%%%%%%%%%%%%%%%%%%%%%%%%%%%%%%%%%%%%%%%%%%%%%%%%%%%%%%%%%%%%%%%%%
%%%%%%%%%%%%%%%%%%%%%%%%%%%%%%%%%%%%%%%%%%%%%%%%%%%%%%%%%%%%%%%%%%%%%%%%%%%%%%%%%%%%%%

\section{Conclusion}

The paper shows the importance of incorporating rotor dynamics in the design of attitude tracking controllers for aerobatic helicopters. It classifies the possible uncertainties associated with the rotor-fuselage model of a helicopter into structured and unstructured disturbances. The proposed BRC controller is robust with respect to both these uncertainties and the tracking error is shown to be ultimately bounded. The ultimate bound of the tracking error can be made arbitrarily small by an appropriate choice of the design parameters $\epsilon_f$ and $\epsilon_r$, being only restricted by control input saturation. The only issue with this controller is its difficulty in implementation due to the need for flap angle feedback. On the other hand, with the knowledge of a few parameters (Table \ref{tab:heli_params}), the easily implementable structure preserving controller can be used to perform aggressive rotational maneuvers at the operational limit of the vehicle, as shown in Sec. \ref{sec:Exp}. On the downside, due to the passive nature of robustness, this controller cannot suppress the error in tracking, arising from parametric uncertainty, to arbitrarily small values. But for all practical purposes the tracking error is guaranteed to be bounded because of the way disturbance enters the system. This controller is also shown to be almost globally asymptotically stable, which is the best that a system defined on a non-Euclidean space can achieve. To the best of the authors' knowledge, this work is the first systematic attempt at designing a globally defined robust attitude tracking controller for an aerobatic helicopter which fully utilizes the rotor dynamics and incorporates the principal uncertainties involved.

%%%%%%%%%%%%%%%%%%%%%%%%%%%%%%%%%%%%%%%%%%%%%%%%%%%%%%%%%%%%%%%%%%%%%%%%%%%%%%%%%%%%%%
%%%%%%%%%%%%%%%%%%%%%%%%%%%%%%%%%%%%%%%%%%%%%%%%%%%%%%%%%%%%%%%%%%%%%%%%%%%%%%%%%%%%%%
%%%%%%%%%%%%%%%%%%%%%%%%%%%%%%%%%%%%%%%%%%%%%%%%%%%%%%%%%%%%%%%%%%%%%%%%%%%%%%%%%%%%%%

%\appendix
%\section{Linearization of error dynamics}  \label{ap:lin}
\appendices
\section{Linearization of error dynamics}  \label{ap:lin}
\begin{proof}
To linearize \eqref{eq:err_dym1_kin}, we first evaluate
\begin{equation}
\frac{d}{d\epsilon} \biggr\rvert_{\epsilon = 0} R(\epsilon) = R_{eq} e^{\epsilon \hat{\bar{\eta}}} \hat{\bar{\eta}} \biggr\rvert_{\epsilon = 0} = R_{eq} \hat{\bar{\eta}}.
\end{equation} 
Linearization of \eqref{eq:err_dym1_kin} is done by taking derivative of LHS and RHS w.r.t $\epsilon$ at $\epsilon = 0$
\begin{equation}
\frac{d}{dt} \frac{d}{d\epsilon} \biggr\rvert_{\epsilon = 0} R(\epsilon) = \frac{d}{d\epsilon} \biggr\rvert_{\epsilon = 0} R(\epsilon) \epsilon \hat{\bar{\omega}}
\end{equation}
which leads to 
\begin{equation}
R_{eq} \dot{\hat{\bar{\eta}}} = R_{eq} \hat{\bar{\omega}} \implies \dot{\bar{\eta}} = \bar{\omega}.
\end{equation}
Using the fact $(Q-Q^T)^\vee = \sum_{i=1}^3 e_i \times Q e_i$ for all $Q \in \mathbb{R}^{3 \times 3}$, $e_{Rm} = \frac{1}{2} \sum_{i=1}^3 e_i \times PR(\epsilon) e_i$. Thus
\begin{equation}
\begin{split}
\frac{d}{d\epsilon} \biggr\rvert_{\epsilon = 0} e_{Rm} = \frac{1}{2}\sum_{i=1}^3 \hat{e}_i P R_{eq} \hat{\bar{\eta}} e_i = -\frac{1}{2}\sum_{i=1}^3 \hat{e}_i P R_{eq} \hat{e}_i \bar{\eta} \\
 = B(R_{eq})\bar{\eta}.
\end{split}
\end{equation}
Using the above relation a similar procedure is followed to linearize \eqref{eq:err_dym1_fus} and \eqref{eq:err_dym1_rot}.
\end{proof}

\nocite{*}% Show all bib entries - both cited and uncited; comment this line to view only cited bib entries;
\IEEEtriggeratref{23}
\bibliographystyle{IEEEtran}
\bibliography{ref}%

\end{document}